\documentclass[journal]{IEEEtran}
\usepackage{hyperref}
\usepackage{cite}

\usepackage{multirow}
\usepackage{graphicx}
\usepackage{amsthm,amsmath,indentfirst,amssymb}
\usepackage{booktabs}
\usepackage{bm}
\newtheorem{definition}{\em Definition}
\newtheorem{lemma}{\em Lemma}
\newtheorem{observation}{\em Observation}
\newtheorem{theorem}{\em Theorem}

\newtheorem{corollary}{\em Corollary}
\newtheorem{claim}{Claim}
\usepackage{algorithmic}
\usepackage{algorithm}
\usepackage{color,xcolor}

\usepackage{array}

\ifCLASSOPTIONcompsoc
\usepackage[caption=false,font=normalsize,labelfont=sf,textfont=sf]{subfig}
\else
\usepackage[caption=false,font=footnotesize]{subfig}
\fi
\usepackage{fixltx2e}

\usepackage{stfloats}
\usepackage{float}
\ifCLASSOPTIONcaptionsoff
  \usepackage[nomarkers]{endfloat}
 \let\MYoriglatexcaption\caption
 \renewcommand{\caption}[2][\relax]{\MYoriglatexcaption[#2]{#2}}
\fi
\usepackage{url}

\begin{document}
%
\title{A Game-Theoretic Approach to Solving the Roman Domination Problem}
%
%
%
\author{Xiuyang~Chen, Changbing~Tang \ {\em Member, IEEE}, Zhao~Zhang \ {\em Member, IEEE}, and Guanrong~Chen \ {\em Life Fellow, IEEE}
	\thanks{This research work is supported in part by National Natural Science Foundation of China (U20A2068), and Zhejiang Provincial Natural Science Foundation of China (LD19A010001).}
	\thanks{X. Chen is with the College of Mathematics and System Science, Xinjiang University, Urumqi, Xinjiang, 830000, China (e-mail: xiuyangchen@126.com). }
	\thanks{C. Tang is with the College of Physics and Electronic Information Engineering, Zhejiang Normal University, Jinhua 321004, China (e-mail: tangcb@zjnu.edu.cn).}
	\thanks{Z. Zhang is with the College of Mathematics and Computer Science, Zhejiang Normal University, Jinhua, Zhejiang, 321004, China (e-mail: hxhzz@sina.com). }
	\thanks{G. Chen is with the Department of Electrical Engineering, City University of Hong Kong, Kowloon, Hong Kong, 999077, China (e-mail: eegchen@cityu.edu.hk).}
	\thanks{Corresponding author: Changbing Tang and Zhao Zhang.}
}

%
%

\markboth{Journal of \LaTeX\ Class Files,~Vol.~XX, No.~XX, XX~XX}%
{Chen \MakeLowercase{\textit{et al.}}:A Game Theoretic Approach for Roman Domination Problem}
%



\maketitle

\begin{abstract}
The Roamn domination problem is one important combinatorial optimization problem that is derived from an old story of defending the Roman Empire and now regains new significance in cyber space security, considering backups in the face of a dynamic network security requirement. In this paper, firstly, we propose a Roman domination game (RDG) and prove that every Nash equilibrium (NE) of the game corresponds to a strong minimal Roman dominating function (S-RDF), as well as a Pareto-optimal solution. Secondly, we show that RDG is an exact potential game, which guarantees the existence of an NE. Thirdly, we design a game-based synchronous algorithm (GSA), which can be implemented distributively and converge to an NE in $O(n)$ rounds, where $n$ is the number of vertices. In GSA, all players make decisions depending on the local information. Furthermore, we enhance GSA to be enhanced GSA (EGSA), which converges to a better NE in $O(n^2)$ rounds. Finally, we present numerical simulations to demonstrate that EGSA can obtain a better approximate solution in promising computation time compared with state-of-the-art algorithms.
\end{abstract}


%
\IEEEpeerreviewmaketitle

\section{Introduction}
%
%
%
%
\IEEEPARstart{T}{he} {\em minimum Roman domination} (MinRD) problem was originated from an interesting historical story \cite{Stewart}: to maintain the safety of the Roman Empire, emperor Constantine adopted the strategy of ``island-hopping"---moving troops from one island to a nearby island, but only when he could leave behind a large enough garrison to keep the first island secure. This strategy was also adopted by MacArthur during the military operations in World War II. In the language of graph theory, MinRD asks for a deployment of the minimum number of troops on some vertices such that any vertex without a troop must be adjacent with a vertex with at least two troops.

MinRD also has contemporary applications, especially in the field of server placements \cite{Pagourtzis} and wireless sensor networks \cite{Ghaffari} in the face of a dynamic security setting. For example, in a wireless sensor network, a set of sensors are deployed at some nodes to monitor the neighboring environment. When a vacant node faces a security problem that needs to be fixed by moving a neighboring sensor to it, the moving sensor should leave behind at least one backup sensor. In other words, a node without a sensor has to be adjacent with a node deployed of at least two sensors. From an economic point of view, it is desired that the total number of sensors is as small as possible, under the condition that the above requirement is satisfied. This consideration leads to a MinRD problem.

The Roman domination problem was mathematically introduced in \cite{Cockayne}. Before that, the definition of a {\em Roman dominating function} (RDF) was given implicitly in \cite{ReVelle} and \cite{Stewart}. MinRD was proved to be NP-hard on general graphs \cite{Dreyer} by a polynomial reduction from the $3$-Satisfiability problem. In \cite{Pagourtzis}, a $(2+2\ln n)$-approximation algorithm is proposed for MinRD on general graphs, and a polynomial-time approximation scheme (PTAS) is designed for MinRD on planar graphs. In \cite{Shang}, it is proved that MinRD is NP-hard even on unit disk graphs, with a $5$-approximation algorithm and a PTAS developed, making use of geometry of unit disk graphs. In \cite{Wang}, a generalized Roman domination problem called {\em connected strong $k$-Roman dominating set problem} is formulated, which is proved to be NP-hard on unit ball graphs, and a $6(k+2)$-approximation algorithm is designed for unit ball graphs. There are also other variants of Roman domination problems \cite{Abdollahzadeh,Ghaffari-Hadigheh,Hajjari}, and many studies of the MinRD for special graphs \cite{Chen1,Ouldrabah,Pavlic,Reddappa,Cockayne1}.

Most algorithms developed in the above articles are centralized, that is, there is a central controller that manages all the processes. Such algorithms are typically vulnerable to cyber attacks. Furthermore, centralized algorithms are difficult to meet the flexibility and diversity requirements of real applications in large-scale networks. Therefore, distributed algorithms are more desirable, especially in multi-agent systems.

In a multi-agent system, every agent can make his own decision using current local information. The autonomy of agents eliminates the dependence on a central controller and greatly improves the anti-attack ability of the system. However, individual interests may have conflict with social welfare, and thus a distributed algorithm may result in an unsatisfactory solution. Game theory is an effective method to coordinate such a conflict. In game theory, it is assumed that all players are selfish, rational and intelligent, who are only interested in maximizing their own benefits. To align individual interests with social welfare, a crucial task is to set up suitable utility functions for the players such that their selfish behaviors can autonomously evolve into a satisfactory and stable collective behavior,  where the stability means that no player is willing to change his current strategy unilaterally. Such a stable state is called a {\em Nash equilibrium} (NE).

With the rapid development of multi-agent systems and large-scale networks, game theory has been widely used in the study of combinatorial optimization problems, for example to deal with the dominating set problem and its variants. In \cite{Yen}, a multi-domination game is studied and it is proved that every NE of the game is a minimal multi-dominating set, which is also a Pareto-optimal solution. In this game, a distributed algorithm is designed, where all players make decisions in a given order. Then, in \cite{Yen2}, an independent domination game is formulated and it is proved that every NE is a minimal independent dominating set. Similarly, all players are required to make decisions only in a given order. Further, in \cite{Chen} a connected domination game is investigated and it is proved that every NE is a minimal connected dominating set but not a Pareto-optimal solution. Later, in \cite{ChenSubmitted}, a secure dominating game is studied and it is proved that every NE is a minimal secure dominating set which is also a Pareto-optimal. Furthermore, a distributed algorithm is designed, which allows all players to make decisions simultaneously and all players use only local information. However, distributed algorithms for the Roman domination problem are rare and many studies on this problem are for special graphs.
The present paper is perhaps the first one using game theory to study the Roman domination problem on general graphs. We propose a Roman domination game (RDG) and prove that every NE of the game is not only a {\em minimal Roman dominating function} (M-RDF) but also {\em strong minimal} (S-RDF), which possesses a local optimum property in a strong sensor. We design a {\em game-based synchronous algorithm} (GSA), which allows the players to make decisions simultaneously, all using local information. Furthermore, we enhance the GSA to the enhanced GSA (EGSA), which converges to a better NE than the GSA.

Another problem closely related to the subject of this paper is the {\em minimum vertex cover } (MinVC) problem. In \cite{Yang}, the MinVC is studied from the approach of snowdrift game and a distributed algorithm is proposed. In \cite{Tang}, the {\em minimum weight vertex cover } (MinWVC) problem is investigated and an algorithm is designed based on an asymmetric game, which can find a vertex cover with smaller weight. In \cite{Sun,Sun3}, the MinWVC problem is solved using potential game theory and a distributed algorithm is proposed based on relaxed greedy algorithm with finite memory. Later, in \cite{Sun2} a population based game theoretic optimizer is designed, which combines learning with optimization. In \cite{Chen2}, a weighted vertex cover
game is proposed using a 2-hop adjustment scheme, so as to obtain a better
solution. In addition to the above works, there are some reports using the game theory to study the coverage problems from different perspectives. For example, in \cite{Li,Li2} cost sharing and strategy-proof mechanisms are adopted for set cover games. In \cite{Fang}, the core stability of vertex cover game is studied. In \cite{Velzen}, the core of a dominating set game in studied. In \cite{Kim}, the core of a connected dominating set game {\color{blue} is} investigated. In \cite{Ai}, the price of anarchy and the computational complexity of a coverage game are analyzed.

In this paper, we focus on the MinRD problem in a multi-agent system. Main contributions of our paper are as follows:
\begin{itemize}
	\item We construct a game framework of multi-agent systems for the MinRD problem and prove the
	existence of NE for the RDG. Furthermore, we classify three types of RDF as $S_{\operatorname{M-RDF}}$, $S_{\operatorname{S-RDF}}$ and $S_{\operatorname{G-RDF}}$, show that $S_{\operatorname{G-RDF}}\subseteq S_{\operatorname{ENE}}\subseteq S_{\operatorname{NE}}\subseteq S_{\operatorname{S-RDF}}\subseteq S_{\operatorname{M-RDF}}\subseteq S_{\operatorname{RDF}}$, and prove that every NE is a Pareto-optimal solution, where $S_{\operatorname{ENE}}$ and $S_{\operatorname{NE}}$ are the set of enhanced NEs and the set of NEs, respectively.
	
	\item We propose three algorithms for the RDG, named {\em game-based asynchronous algorithm} (GAA), {\em game-based synchronous algorithm} (GSA) and  {\em enhanced game-based synchronous algorithm} (EGSA). In GAA, an NE converges in $O(n)$ rounds of interactions and players make decisions depending on local information at most one hop away. In GSA, not only can it achieve the same performances as GAA, but also the computation can be realized distributedly, where all players make decisions simultaneously. In EGSA, a better solution than GSA can be obtained, at the expense of a longer running time $O(n^2)$.
	\item Our numerical simulation shows that the proposed algorithms are better than the existing algorithms on random graphs. On random tree graphs, our algorithms yield results closer to the optimal solutions, whereas the gap between the result of our EGSA and the optimal solution is less than 0.01\%.

\end{itemize}

The remaining parts of the paper are organized as follows. Section $2$ introduces preliminaries in game theory and MinRD, and constructs an RDG. Section $3$ provides a strict theoretical analysis for the game and explores the relationship between NE and the three types of RDF. Section $4$ describes the three algorithms in details and provides theoretical analysis on its convergence. Section $5$ evaluates the performance of the three algorithms through extensive simulations. Section $6$ concludes the paper with some discussions on future work.



\section{Modelling Roman Domination Problem as a Game}\label{sec2}

\subsection{Preliminaries}

\begin{table}[tbp]
	\centering
	\caption{List of Notations}
	\label{[Tab:01]}
	\scalebox{0.9}{
	 \begin{tabular}{|c|l|}
		\hline
		Notation&Meaning\\
		\hline
		\ $\Gamma$ & The game\\
		\ $n$ & The number of vertices, also the number of players\\
		\ $v_i$ & Both a vertex in the graph and a player in the game\\
		\ $V$ & The set of players or vertices $V=\{v_1,v_2,\ldots,v_n\}$\\
		\ $S_i$ & $S_i=\{0,1,2\}$ is the strategy set of player $v_i$\\
		\ $\Sigma$ & $\Sigma=S_1\times \cdots \times S_n$ is the strategy space\\
		\ $C$ & $C=(c_1,c_2,\ldots,c_n)\in\Sigma$ is a strategy profile\\
		\ $d(i,j)$ & The length of a shortest path between vertex $v_i$ and $v_j$\\
		\ $N_i$ & $N_i=\{v_j\in V \colon v_j$ is adjacent to $v_i\}$ is the neighbor set of $v_i$\\
		\ $\bar{N_i}$ & $\bar{N_i}=N_i\cup\{v_i\}$ is the closed neighbor set of $v_i$\\
		\ $N_{i,k}$ & $N_{i,k}=\{v_j\in V: 0<d(i,j)\leq k\}$ is the $k$-hop neighbor set of $v_i$\\
		\ $\bar N_{i,k}$ & $\bar N_{i,k}=N_{i,k}\cup\{v_i\}$ is the closed $k$-hop neighbor set of $v_i$\\
		\ $N_i^j(C)$ & $N_i^j(C)=\{v_k\in N_i\colon c_k=j\}$\\
		\hline
	\end{tabular}
}
\end{table}

A game $\Gamma$ is represented by $\Gamma = (V,\{S_i\}_{i=1}^{n},\{u_i\}_{i=1}^{n})$, where $V=\{v_1,v_2,\ldots,v_n\}$ is the set of {\em players}, $S_i$ is $v_i$'s {\em strategy set}, and $u_i$ is $v_i$'s {\em utility function}. The {\em strategy space} of the game is $\Sigma=S_1\times S_2\times\cdots\times S_n$. A {\em strategy profile} is an $n$-tuple $C=(c_1,c_2,\ldots,c_n)\in\Sigma$. For player $v_i$, write $C=(c_i,C_{-i})$, where $C_{-i}=(c_1,\ldots,c_{i-1},c_{i+1},\ldots,c_n)$ indicates the strategies of those players except $v_i$. Let $u_i(C)$ denote the utility of $v_i$ under strategy profile $C$. Players are assumed to be selfish, intelligent and rational, which means that the goal of every player tries to maximize his own utility. The {\em best response} of player $v_i$ to the current strategy profile $C$ is
$$
\mbox{BR}(v_i,C)=\arg\max\{u_i(c_i',C_{-i})\colon c_i'\in S_i\}.
$$
A {\em Nash equilibrium} (NE) is a strategy profile $C$ such that no player wants to deviate from $C$ unilaterally, formally defined as follows.

\begin{figure*}[tbp]
	\centering
	\subfloat[$f^{(1)}$: an RDF]{\includegraphics[width=0.2\linewidth]{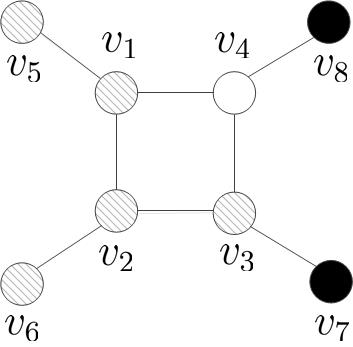}}\hfil
	\subfloat[$f^{(2)}$: an M-RDF]{\includegraphics[width=0.2\linewidth]{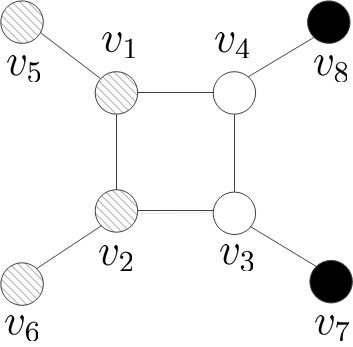}}\hfil
	\subfloat[$f^{(3)}$: an S-RDF ]{\includegraphics[width=0.2\linewidth]{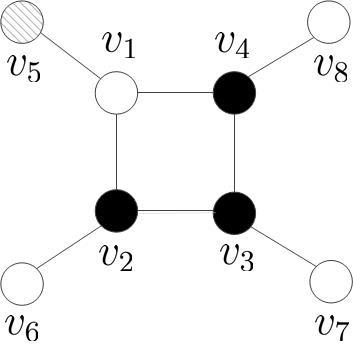}}\hfil
	\subfloat[$f^{(4)}$: a G-RDF]{\includegraphics[width=0.2\linewidth]{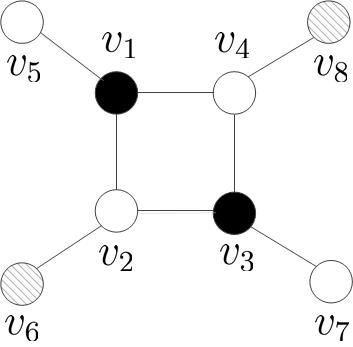}}
	\caption{An example to show the four types of RDFs.}
	\label{fig1}
\end{figure*}

\begin{definition}[Nash equilibrium \cite{Nash}]
	{\rm Given a game $\Gamma = (V,\{S_i\}_{i=1}^{n},\{u_i\}_{i=1}^{n})$, a strategy profile
		$C^{*}=(c_1^{*},c_2^{*},\ldots,c_n^{*})$ is a {\em Nash equilibrium} (NE) if
		$$
		u_i(c_i^{*},C_{-i}^*)\geq u_i(c_i,C_{-i}^*)\ \mbox{for any}\ v_i\in V\ \mbox{and any}\ c_i\in S_i.
		$$}
\end{definition}

Notice that an NE is not necessarily a global optimal solution. Especially, for an NP-hard problem, in many cases, Pareto-optimal solutions are satisfactory.

\begin{definition}[Pareto-optimal solution]
	{\rm Given a game $\Gamma = (V,\{S_i\}_{i=1}^{n},\{u_i\}_{i=1}^{n})$, a strategy profile $C'=(c_1',\ldots,c_n')$ {\em strictly dominates} strategy profile $C=(c_1,\ldots,c_n)$ if
		$u_i(C')\geq u_i(C)$ holds for all index $i\in\{1,\ldots,n\}$ and there exists an index $j\in\{1,\ldots,n\}$ with $u_j(C')>u_j(C)$. A strategy profile $C^*=(c_1^*,\ldots,c_n^*)$ is a {\em Pareto-optimal solution} if there is no strategy profile which strictly dominates $C^*$.}
\end{definition}

\begin{definition}[Roman dominating function (RDF)]\label{DefRDF}
	{\rm Given a graph $G=(V,E)$, a function $f:V(G)\rightarrow\{0,1,2\}$ is an RDF of $G$ if each vertex $v\in V$ with $f(v)=0$ is adjacent to a vertex $u\in V$ with $f(u)=2$.}
\end{definition}

\begin{definition}[Minimum Roman domination problem (MinRD)]
{\rm Given a graph $G=(V,E)$, for any RDF $f$, let $\gamma_{R}(G,f)=\sum_{v\in V}f(v)$ be the weight of the RDF $f$. The goal of MinRD is to minimize $\gamma_{R}(G,f)$ over all RDFs $f$.}
\end{definition}

Denote by $S_{RDF}$ the set of RDFs and define three types of RDF as follows, in terms of their qualities.

\begin{definition}[Global minimum Roman dominating function (G-RDF)]\label{DefGRDF}
	{\rm An RDF $f^*$ of graph $G$ is a G-RDF if $\gamma_R(G,f^*)=\min_{f\in S_{\operatorname{RDF}}}\gamma_R(G,f)$. Let $S_{\operatorname{G-RDF}}$ be the set of G-RDFs.}
\end{definition}

\begin{definition}[Minimal Roman dominating function (M-RDF)]\label{DefMRDF}
	{\rm An RDF $f$ of graph $G$ is an M-RDF if there is no RDF $f'$ of $G$ satisfying $f'(v_i)=f(v_i)-1$ for some $v_i\in V$ and $f'(v_j)=f(v_j)$ for any $j\neq i$. Let $S_{\operatorname{M-RDF}}$ be the set of M-RDFs.}
\end{definition}

\begin{definition}[Strong minimal Roman dominating function (S-RDF)]\label{DefSRDF}
	{\rm An M-RDF $f$ is an S-RDF if there is no RDF $f'$ with $\gamma_R(G,f')<\gamma_R(G,f)$ satisfying $f'(v_i)=2$ for some $v_i\in V$, $f'(v_j)=0$ for any $v_j\in N_i$ with $f(v_j)=1$  and $f'(v_k)=f(v_k)$ for any other $v_k$. Let $S_{\operatorname{S-RDF}}$ be the set of S-RDFs.}
\end{definition}

Fig. \ref{fig1} shows four different types of RDFs. For simplicity of statement, call a vertex $v_i$ to be {\em white, gray, or black} if $f(v_i)=0,1$ and $2$, respectively. In Fig. \ref{fig1}, $f^{(1)},f^{(2)},f^{(3)}$ and $f^{(4)}$ are all RDFs, because
every white vertex is adjacent to at least one black vertex. But $f^{(1)}$ is not an M-RDF, because $f'$ defined in the following way is also an RDF: $f'(v_3)=f^{(1)}(v_3)-1$ and $f'(v_j)=f^{(1 )}(v_j)$ for any $v_j\neq v_3$. It can be checked that $f^{(2)},f^{(3)}$ and $f^{(4)}$ are all M-RDFs. For example, in Fig. \ref{fig1} $(b)$, after changing a gray vertex $v_i$ with $i\in\{1,2,5,6\}$ to a white vertex, there is no black vertex adjacent to $v_i$; after changing a black vertex $v_7$ (or $v_8$) to a gray vertex, there is no black vertex adjacent to $v_3$ (or $v_4$). In both cases, $f'$ is no longer an RDF. Note that  $f^{(2)}$ is not an S-RDF, because $f'$ defined in the following way is an RDF with $\gamma_R(G,f')=7<8=\gamma_R(G,f^{(2)})$:  $f'(v_1)=2$, $f'(v_j)=0$ for $j=2,5$ and $f'(v_k)=f(v_k)$ for $k=3,4,6,7,8$. It can be checked that both $f^{(3)}$ and $f^{(4)}$ are S-RDFs. Furthermore, $f^{(4)}$ is a G-RDF, while $f^{(3)}$ is not.

To facilitate reading, main notations are summarized in Table \ref{[Tab:01]}.

\subsection{Roman Domination Game}\label{sec2.2}
In this subsection, the RDG is introduced. Given a graph $G=(V,E)$, each vertex $v\in V$ can be viewed as a player. For a player $v_i$, his strategy $c_i=0,1$ or $2$ indicates that $f(v_i)=0,1$ or $2$, respectively.

It should be noted that RDF is {\em monotonic} in the sense that if $C$ is a profile corresponding to an RDF, $C'$ is another profile with $C'\geq C$ (that is $c'_i\geq c_i$ for every $i$), thus $C'$ also corresponds to an RDF; and if $C$ is not an RDF, $C'\leq C$ is not an RDF either.

Assuming that all players are selfish, intelligent and rational, they will not consider the benefits of the other players while seeking to maximize their own benefits. The critical task of the game is to design a good utility function for the players, such that a stable and fairly good social state can be reached through cooperation and competition among players, where the goodness of the state is measured by the following criteria.

\begin{itemize}
	\item[$(\romannumeral1)$] {\em Self-Stability:} Starting from any initial state, the game can end up in a Nash equilibrium which corresponds to an RDF.
	\item[$(\romannumeral2)$] {\em Solution Quality:} The weight of the RDF corresponding to the NE should be reasonably small. Since the computation of a G-RDF is NP-hard even using centralized algorithms, one cannot hope for a minimum solution of MinRD in reasonable time. An alternative basic requirement is that the computed RDF should be minimal, that is, decreasing one bit of the $f$-value of any vertex $v_i$ with $f(v_i)\neq0$ will no longer be an RDF. Furthermore, it is desired to maximize social welfare, for which Pareto optimality is an important indicator.
	
	\item[$(\romannumeral3)$] {\em Efficient Execution:} The time for the game to reach an NE should be polynomial in the size of the input. The information used for players should be local. Moreover, a distributed algorithm is preferred.
\end{itemize}

For a strategy profile $C=(c_1,c_2,\ldots,c_n)$, we define the utility function of $v_i$ as
\begin{equation}
	u_i(C)=g_i(C)+q_i(C),
\end{equation}
with
\begin{equation}
	g_i(C)=-\lambda_1c_i^2,
\end{equation}
\begin{equation}
	q_i(C)=-\lambda_2\sum_{v_j\in \bar{N_i}}(2-c_j)m_j(C),
\end{equation}
where
\begin{equation}
	m_j(C)=\begin{cases}
		1,& \bar{N}_j^2(C)=\emptyset,\\
		0,& {\rm otherwise.}
	\end{cases}
\end{equation}
and $\lambda_1,\lambda_2$ are constants satisfying $0<\frac{2}{3}\lambda_2<\lambda_1<\frac{3}{4}\lambda_2$.

In the following, for a strategy profile $C$, call a vertex $v_i$ as {\em strongly dominated} if there is a vertex $v_j\in\bar{N}_i$ with $c_j=2$, otherwise call $v_i$ as {\em free}. Note that a player
\begin{equation}\label{eq0710-1}
	\mbox{$v_i$ is strongly dominated if and only if $m_i(C)=0$.}
\end{equation}
As can be seen that a black vertex is strongly dominated by itself. A vertex is {\em dominated} if it is either a black, or a gray, or a strongly dominated white vertex. Note that $C$ is an RDF if and only if all vertices are dominated.

Furthermore, assume that $G$ has no isolated vertex. In fact, if $v_i$ is an isolated vertex, in order for $C$ to be an M-RDF, $v_i$ must be gray. From the game point of view, the best response of $v_i$ is also gray, since $u_i(c_i=0,C_{-i})=-2\lambda_2$, $u_i(c_i=1,C_{-i})=-\lambda_1-\lambda_2$ and $u_i(c_i=2,C_{-i})=-4\lambda_1$. So, such trivial case will be ignored in the later discussion.

\section{Theoretical Analysis}\label{sec3}
In this section, the theoretical properties of the RDG designed in the above section will be analyzed. In the following, we shall use $C$ to denote both a strategy profile $(c_1,\ldots,c_n)\in\Sigma$ and the corresponding function $f$ with $f(v_i)=c_i$, $(i=1,\dots,n)$. It is assumed that a player is willing to change his strategy only when he can be {\em strictly} better off, that is, when his utility becomes strictly larger after changing his current strategy.

\subsection{Nash Equilibrium}

Nash equilibrium \cite{Nash} is a stable state in the game. In this subsection, the properties of NE in RDG will be analyzed.

\begin{observation}\label{obs1}	
	By the definition of $m_j(C)$, $(a)$ changing strategy of $v_i$ between 0 and 1 does not affect $m_j(C)$ for any $v_j$; $(b)$ changing gray vertex to white does not make the other vertices be strongly dominated or free.
\end{observation}

\begin{lemma}\label{lem0}
	For any player $v_i$ with $c_i=2$ under strategy profile $C$, $u_i(C)=-4\lambda_1$.
\end{lemma}
\begin{proof}
	Every player $v_j\in \bar{N}_i$ is strongly dominated because $c_i=2$, and thus $m_j(C)=0$ by observation \eqref{eq0710-1}. The lemma follows from the definition of the utility function.
\end{proof}

\begin{lemma}\label{lem5}
	In any strategy profile $C$, for any player $v_i\in V$, $u_i(c_i=1,C_{-i})>u_i(c'_i=0,C_{-i})$ if and only if $m_i(C)=1$.
\end{lemma}
\begin{proof}
	Let $C=(c_i=1,C_{-i})$ and $C'=(c'_i=0,C_{-i})$. Then $$u_i(C)=-\lambda_1-\lambda_2m_i(C)-\lambda_2\sum_{v_j\in N_i}(2-c_j)m_j(C),$$ $$u_i(C')=-2\lambda_2m_i(C')-\lambda_2\sum_{v_j\in N_i}(2-c_j)m_j(C').$$
	By Observation \ref{obs1} $(a)$, $m_j(C)=m_j(C')$, $	\forall v_j\in V $. Thus $u_i(C)-u_i(C')=\lambda_2m_i(C)-\lambda_1>0$ if and only if $m_i(C)=1$ by $\lambda_2>\lambda_1$.
\end{proof}

\begin{theorem}\label{thm1}
	Every NE of the RDG is an RDF.
\end{theorem}
\begin{proof}[Proof]
	Suppose this is not true.
	Let $C$ be an NE which is not an RDF. Then, $\exists v_i\in V$ with $c_i=0$ and $N_i^2(C)=\emptyset$. It follows that $m_i(C)=1$ by the definition of $m_i(C)$.  Let $C'$ be $(c'_i=1,C_{-i})$. By Lemma \ref{lem5}, $u_i(c_i'=1,C_{-i})>u_i(c_i=0,C_{-i})$. It contradicts that $C$ is an NE. Hence, $C$ is an RDF.	
\end{proof}

\begin{theorem}\label{thm2}
	Every NE of an RDG is an M-RDF.
\end{theorem}
\begin{proof}
	Suppose $C$ is an NE, but is not an M-RDF. By Theorem \ref{thm1}, it suffices to consider the minimality of $C$. Suppose $v_i\in V$ is a vertex with $c_i\geq 1$ such that $C'=(c'_i=c_i-1,C_{-i})$ is still an RDF. There are two cases to be considered.
	
	Case $1$: $c_i=1$
	
	In this case, the assumption of $C'=(c'_i=0,C_{-i})$ is still an RDF implies that
	$N_i^2(C')=N_i^2(C)\neq\emptyset$, and thus $m_i(C')=m_i(C)=0$. By Lemma \ref{lem5}, $u_i(c_i'=0,C_{-i})>u_i(c_i=1,C_{-i})$, contradicting to the assumption that $C$ is an NE.
	
	Case $2$: $c_i=2$
	
	In this case, $C'=(c'_i=1,C_{-i})$ is also an RDF. If $N_i^1(C)\neq\emptyset$, let $v_j\in N_i^1(C)$. Then $(c_j'=0,C_{-j})$ is still an RDF by Observation \ref{obs1} $(b)$. This situation is reduced to Case $1$ (with $v_j$ playing the role of $v_i$). Hence, suppose $N_i^1(C)=\emptyset$. The assumption of $C'=(c'_i=1,C_{-i})$ is an RDF implies that for any $v_j\in N_i^0(C')$, there exists a $v_k\in N_j^2(C')$, and thus $m_j(C')=0$ by observation (\ref{eq0710-1}). Hence, $u_i(c_i'=1,C_{-i})=-\lambda_1-\lambda_2m_i(C')\geq -\lambda_1-\lambda_2$ (note that $m_i(C')\leq1$). Because $u_i(c_i=2,C_{-i})=-4\lambda_1$ by Lemma \ref{lem0}, $u_i(c_i'=1,C_{-i})-u_i(c_i=2,C_{-i})\geq3\lambda_1-\lambda_2>\lambda_2>0$ by the assumption of $0<\frac{2}{3}\lambda_2<\lambda_1$, which is a contradiction to that $C$ is an NE.
\end{proof}

\begin{corollary}\label{cor1217-1}
	In any NE $C$, the following three conditions are satisfied:
	\begin{itemize}
		\item[$(\romannumeral1)$] No gray vertices are adjacent to black vertices,
		\item[$(\romannumeral2)$]  Any white vertex $v_i$ is adjacent to at least one black vertex,
		\item[$(\romannumeral3)$] If a vertex $v_i$ is gray, then $m_i(C)=1$.
	\end{itemize}
\end{corollary}

\begin{proof}
	
	Property $(\romannumeral1)$ follows from the minimality of $C$ and Observation \ref{obs1} $(b)$. In fact, if there exists a gray vertex $v_i$ adjacent to a black vertex $v_j$, then $C'=(c_i'=0,C_{-i})$ is also an RDF, contradicting the minimality of $C$.

	Property $(\romannumeral2)$ is a direct result of the definition of RDF and Theorem \ref{thm1}.
	
	Property $(\romannumeral3)$ follows from property $(\romannumeral1)$, because for any gray vertex $v_i$, $N^2_i(C)=\emptyset$, and thus $m_i(C)=1$ by the definition of $m_j(C)$.
\end{proof}

\begin{figure*}[htb]
	\centering
	\subfloat[]{\includegraphics[width=1.5cm,height=2.5cm]{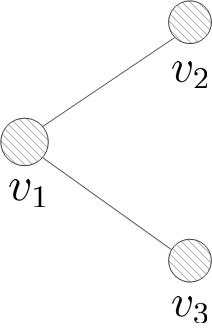}}\hfil
	\subfloat[]{\includegraphics[width=1.5cm,height=2.5cm ]{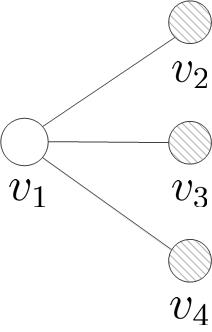}}
	\caption{Some ``bad" substructures $(a)$ A gray vertex having at least two gray neighbors; $(b)$ A white vertex having at least three gray neighbors.}\label{fig2}
\end{figure*}

\begin{theorem}\label{thm3}
	Every NE of an RDG is an S-RDF.
\end{theorem}
\begin{proof}
	Suppose $C$ is an NE. By Theorem \ref{thm2}, $C$ is an M-RDF. So, if $C$ is not an S-RDF, then $\exists$ an RDF $C'$ with $\gamma_R(G,C')<\gamma_R(G,C)$ satisfying $c'_i=2$ for some $v_i\in V$, $c'_j=0$ for any $v_j\in N_i$ with $c_j=1$, and $c'_k=c_k$ for the other $v_k\notin\bar{N_i}$.
	
	If $c_i=2$, then by Corollary \ref{cor1217-1} $(\romannumeral1)$, no vertex in $N_i$ has value 1. So, $C'=C$, a contradiction to $\gamma_R(G,C')<\gamma_R(G,C)$.
	
	If $c_i=1$, then by $\gamma_R(G,C')<\gamma_R(G,C)$, there are at least two vertices $v_j,v_k\in N_i$ with $c_j=c_k=1$ and $c'_j=c'_k=0$. By Corollary \ref{cor1217-1} $(\romannumeral3)$, $m_i(C)=m_j(C)=m_k(C)=1$. Then $u_i(C)\leq-\lambda_1-3\lambda_2$. Consider a strategy profile $C''=(c_i''=2,C_{-i})$. By Lemma \ref{lem0}, $u_i(C'')=-4\lambda_1$. It follows that $u_i(C'')-u_i(C)\geq3\lambda_2-3\lambda_1>0$, contradicting that $C$ is an NE.
	
	If $c_i=0$, then $\gamma_R(G,C')<\gamma_R(G,C)$ implies that, there are at least three vertices $v_j,v_k,v_l\in N_i$ with $c_j=c_k=c_l=1$ and $c'_j=c'_k=c_l'=0$. Similarly as above, $m_j(C)=m_k(C)=m_l(C)=1$, and thus $u_i(C)\leq-3\lambda_2$, while $C''=(c_i''=2,C_{-i})$ has $u_i(C'')=-4\lambda_1$. Hence $u_i(C'')-u_i(C)\geq3\lambda_2-4\lambda_1>0$, contradicting that $C$ is an NE.
\end{proof}

In the above proof, one can see that the two substructures in Fig. \ref{fig2} do not appear in an NE.

\begin{corollary}\label{cor3}
	No NE of an RDG has the two ``bad" substructures shown in Fig. \ref{fig2}.
\end{corollary}

In fact, the two substructures do not appear in any S-RDF.

\begin{lemma}\label{lem1}
	For any S-RDF $C$, the two substructures shown in Fig. \ref{fig2} cannot appear.
\end{lemma}
\begin{proof}
	Define
	\begin{equation}
		C'=\begin{cases}
			2,& \mbox{for}\ v_1,\\
			0,& \mbox{for}\ v_j\in N_1\wedge (c_j\neq2),\\
			c_j,& \mbox{for other}\ v_j.
		\end{cases}\nonumber
	\end{equation}
	Then, for both cases in Fig. \ref{fig2}, it can be checked that $C'$ is also an RDF and $\gamma_R(G,C')\leq\gamma_R(G,C)-1<\gamma_R(G,C)$, and the lemma follows from the definition of S-RDF.
\end{proof}

By Theorem \ref{thm3}, $S_{\operatorname{NE}}\subseteq S_{\operatorname{S-RDF}}$, where $S_{\operatorname{NE}}$ is the set of NEs of the RDG. Next, it can be proved that  $S_{\operatorname{G-RDF}}\subseteq S_{\operatorname{NE}}$. This is a result of Lemma \ref{lem14} (when we study the algorithmic aspect of the game). The following is an existence proof. It is included here for theoretical completeness. For this purpose, first a lemma is established.

\begin{lemma}\label{lem3}
	In any S-RDF, there is no player $v_i\in V$ who is willing to increase his $c_i$-value unilaterally.
\end{lemma}
\begin{proof}
	Let $C$ be an S-RDF. If $c_i=2$, then $v_i$ cannot increase its value.
	
	Suppose $v_i$ increases its value from $c_i=1$ to $c_i'=2$. Because $C$ is an S-RDF, every white vertex is adjacent to at least one black vertex. So, $m_j(C)=0$ for any $v_j$ with $c_j=0$. Hence, $u_i(C)=-\lambda_1-\lambda_2\sum_{v_j\in \bar N_i^1(C)}(2-c_j)m_j(C)$. By Lemma \ref{lem1}, there is at most one gray vertex $v_k\in N_i^1(C)$. Hence, $u_i(C)\geq -\lambda_1-2\lambda_2$. Since $u_i(c_i'=2,C_{-i})=-4\lambda_1$ by Lemma \ref{lem0}, $u_i(C)-u_i(c_i'=2,C_{-i})\geq 3\lambda_1-2\lambda_2>0$.
	
	Suppose $v_i$ increases its value from $c_i=0$ to $c_i'=1$. Because $C$ is an S-RDF, $N_i^2(C)\neq\emptyset$, and thus $m_i(C)=0$. By Lemma \ref{lem5}, $u_i(C)>u_i(c_i'=1,C_{-i})$.
	
	Suppose $v_i$ increases its value from $c_i=0$ to $c_i'=2$. Because $C$ is an S-RDF, $m_j(C)=0$ for any $v_j$ with $c_j=0$. By Lemma \ref{lem1}, there are at most two gray vertices in $N_i^1(C)$. So, $u_i(C)\geq -2\lambda_2$. Since $u_i(c_i'=2,C_{-i})=-4\lambda_1$ by Lemma \ref{lem0}, $u_i(C)-u_i(c_i'=2,C_{-i})\geq 4\lambda_1-2\lambda_2>0$.
	
	In any case, $v_i$ is not willing to increase his strategy from $c_i$ to $c_i'$.
\end{proof}

Note that Lemma \ref{lem3} implies that in an S-RDF, no player is willing to ``increase'' value, but ``decreasing'' value is still possible.

The next lemma gives some properties of G-RDF, which will be used to show that any G-RDF is an NE. For a strategy profile $C$, a vertex $v_j$ is {\em uniquely strongly dominated} by $v_i$ if $v_i$ is the unique vertex in $\bar N_j$ with $c$-value 2. Note that this definition allows $v_i$ to be $v_j$ itself.

\begin{lemma}\label{lem0709-1}
	Let $C$ be a $\operatorname{G-RDF}$, the following two properties hold:
	\begin{itemize}
		\item[$(\romannumeral1)$] No gray vertex is adjacent to a black vertex.
		\item[$(\romannumeral2)$]  For any black vertex $v_i$, there are at least two vertices in $\bar N_i$ which are uniquely strongly dominated by $v_i$.
	\end{itemize}
\end{lemma}
\begin{proof}
	Suppose $C$ has a gray vertex $v_i$ which is adjacent to a black vertex. Then, $v_i$ is still a dominated vertex if its status is changed from gray to white. By Observation \ref{obs1} $(b)$, the status of the other vertices (dominated or not) are not affected by such a change. Since all vertices are dominated in $C$, so are all vertices in $C'=(c'_i=0,C_{-i})$. But, then, $C'$ is an RDF with smaller weight than $C$, contradicting that $C$ is a G-RDF. Property $(\romannumeral1)$ is proved.
	
	Consider a black vertex $v_i$. By $(\romannumeral1)$, all neighbors of $v_i$ are either black or white. Since it has been assumed that $G$ has no isolation vertices,  one has $N_i\neq\emptyset$. If all neighbors of $v_i$ are black, then $C'=(c_i'=0,C_{-i})$ is an RDF with $\gamma_R(G,C')=\gamma_R(G,C)-2$, contradicting that $C$ is a G-RDF. So, $|N_i^0(C)|\geq 1$. If every vertex in $N_i^0(C)$ has another black neighbor, then $C'=(c_i'=1,C_{-i})$ is an RDF with $\gamma_R(G,C')=\gamma_R(G,C)-1$, also a contradiction. So, there is at least one vertex in $N_i^0(C)$, say $v_j$, which is uniquely strongly dominated by $v_i$. If all neighbors of $v_i$ are white, then $v_i$ is a vertex which is strongly dominated by itself, property $(\romannumeral2)$ is satisfied. Next, suppose $v_i$ has a black neighbor. If $v_j$ is the only vertex in $N_i^0(C)$ which is uniquely strongly dominated by $v_i$, then $C'=(c_i'=0,c_j'=1,C_{-\{i,j\}})$ is an RDF with $\gamma_R(G,C')=\gamma_R(G,C)-1$, again a contradiction. So, in this case, there are at least two vertices in $N_i^0(C)$ which are uniquely strongly dominated by $v_i$. Property $(\romannumeral2)$ is proved.
\end{proof}

Now, it is ready to prove $S_{\operatorname{G-RDF}}\subseteq S_{\operatorname{NE }}$.

\begin{theorem}\label{thm4}
	In an RDG, every G-RDF is an NE.
\end{theorem}
\begin{proof}
	Suppose $C$ is a G-RDF. Because a G-RDF is also an S-RDF, by Lemma \ref{lem3}, no player is willing to increase his $c$-value. So, to show that $C$ is an NE, it suffices to show that no player $v_i$ is willing to decrease his $c$-value from $c_i$ to $c_i'<c_i$.
	
	If $c_i=2$, then $u_i(C)=-4\lambda_1$, and by Lemma \ref{lem0709-1} $(\romannumeral2)$, there are at least two vertices $v_j,v_k\in \bar N_i$ which are uniquely strongly dominated by $v_i$. Reducing $c_i=2$ to $c_i'=1$ or $c_i'=0$ results in a strategy profile $C'$. Note that $v_j$ and $v_k$ are no longer strongly dominated in $C'$, and thus $m_j(C')=m_k(C')=1$. If $v_i\in\{v_j,v_k\}$, then the other vertex in $\{v_j,v_k\}$ is white by Lemma \ref{lem0709-1} $(\romannumeral1)$, and thus $u_i(c_i'=1,C_{-i})\leq -\lambda_1-3\lambda_2$, $u_i(c_i'=0,C_{-i})\leq -4\lambda_2$. If $v_i\not\in \{v_j,v_k\}$, then both $v_j,v_k$ are white, and thus $u_i(C')\leq-4\lambda_2$. In any case, $u_i(C)>u_i(C')$, and $v_i$ is not willing to change.
	
	If $c_i=1$, then by Lemma \ref{lem0709-1} $(\romannumeral1)$, $\bar{N}_i^2(C)=\emptyset$. It follows that $m_i(C)=1$. So, $u_i(C)>u_i(c_i'=0,C_{-i})$ by Lemma \ref{lem5}, and thus $v_i$ is not willing to change.
\end{proof}

By Definitions \ref{DefRDF}, \ref{DefMRDF} and \ref{DefSRDF}, $S_{\operatorname{S-RDF}}\subseteq S_{\operatorname{M-RDF}}\subseteq S_{\operatorname{RDF}}$. Combining this with Theorem \ref{thm3} and Theorem \ref{thm4}, one has the following corollary.

\begin{corollary}\label{cor1}
	In an RDG, $S_{\operatorname{G-RDF}}\subseteq S_{\operatorname{NE}}\subseteq S_{\operatorname{S-RDF}}\subseteq S_{\operatorname{M-RDF}}\subseteq S_{\operatorname{RDF}}.$
\end{corollary}

The relationship among these sets is illustrated in Fig. \ref{fig3}.

\begin{figure}[h]
	\centering
	\includegraphics[width=0.45\linewidth]{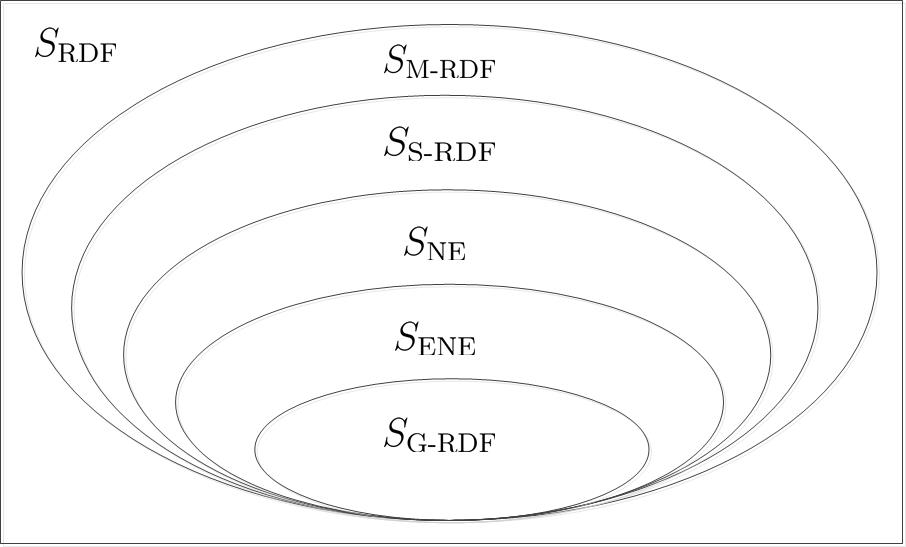}
	\caption{Relationship among different types of RDF, where $S_{\operatorname{ENE}}$ is the set of ENEs obtained by Algorithm \ref{alg3} defined in Section \ref{sec4} }
	\label{fig3}
\end{figure}

From Fig. \ref{fig1}, it can be seen that $S_{\operatorname{S-RDF}}$ might be a proper subclass of $S_{\operatorname{M-RDF}}$. Next, it will be shown that the difference can be very large. For this purpose, define a parameter $\rho_n$ as follows, as a measure of how big the difference is.

\begin{definition}
	Let $\mathbb{G}$ be the set of all graphs on $n$ vertices. Define $$ \rho_n=\max_{G\in\mathbb{G}}\frac{\min_{C'\in S_{\operatorname{M-RDF}}\setminus S_{\operatorname{S-RDF}}}\gamma_{R}(G,C')}{\max_{C''\in S_{\operatorname{S-RDF}}}\gamma_{R}(G,C'')}.$$
\end{definition}

Consider the example in Fig. \ref{fig4}. It can be checked that $S_{\operatorname{M-RDF}}\setminus S_{\operatorname{S-RDF}}=\{C_1,C_2\}$, and $S_{\operatorname{S-RDF}}=\{C_3\}$. Hence, $\rho_n\geq \frac{n}{2}$.

\begin{figure*}[h]
	\centering
	\subfloat[$C_1$]{\includegraphics[width=0.2\linewidth]{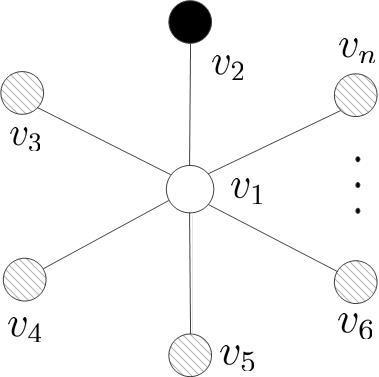}}\hfil\hfil
	\subfloat[$C_2$]{\includegraphics[width=0.2\linewidth]{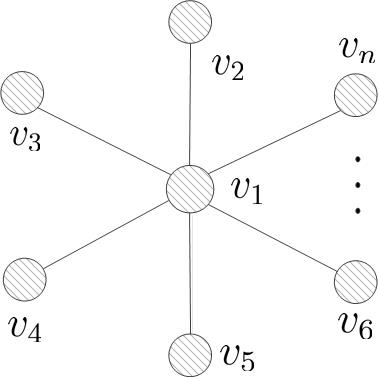}}\hfil\hfil
	\subfloat[$C_3$]{\includegraphics[width=0.2\linewidth]{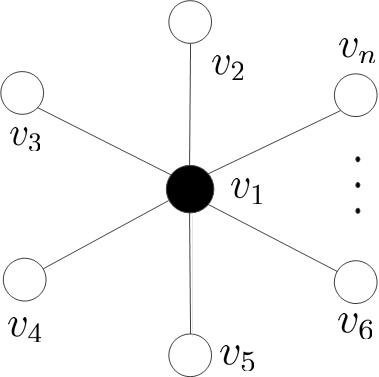}}
	\caption{An example showing that the difference between $S_{\operatorname{M-RDF}}\setminus S_{\operatorname{S-RDF}}$ and $S_{\operatorname{S-RDF}}$ can be very large.}
	\label{fig4}
\end{figure*}

It has been proved that $S_{\operatorname{NE}}\subseteq S_{\operatorname{S-RDF}}$. The example in Fig. \ref{fig5} shows that the inclusion might be proper. In fact, the strategy profile $C$ in Fig. \ref{fig5} $(a)$ satisfies $C\in S_{\operatorname{S-RDF}}\setminus S_{\operatorname{NE}}$, because $u_2(C)=-4\lambda_1<-2\lambda_2=u_2(c_2'=0,C_{-2})$. Note that $C'$ is both an NE and an S-RDF.

\begin{figure*}[h]
	\centering
	\subfloat[$C\in S_{\operatorname{S-RDF}}\setminus S_{\operatorname{NE}}$ and $\gamma_R(G,C)=n$]{\includegraphics[width=0.45\linewidth]{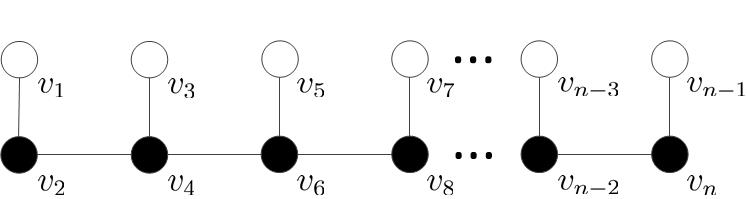}}\hfil
	\subfloat[$C'\in S_{\operatorname{S-RDF}}\cap S_{\operatorname{NE}}$, $\gamma_R(G,C')=\frac{3}{4}n$]{\includegraphics[width=0.45\linewidth]{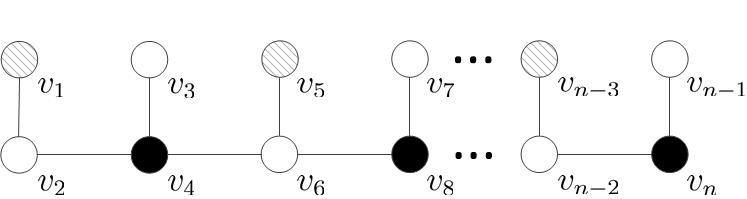}}
	\caption{An example showing that $S_{\operatorname{S-RDF}}\setminus S_{\operatorname{NE}}\neq\emptyset$.}
	\label{fig5}
\end{figure*}

\subsection{Potential Game}

In this subsection, it will be shown that the proposed RDG is an exact potential game, and thus NEs exist. Furthermore, an NE can be reached in linear rounds of interactions if the players are allowed to determine their strategies one by one.

\begin{definition}[exact potential game \cite{Monderer}]
	{\rm Call $\Gamma = (V;\{S_i\}_{i=1}^{n};\{u_i\}_{i=1}^{n})$ an {\em exact potential game} if there exists a potential function $\pi:\Sigma\mapsto\mathbb{R}$ such that for any player $v_i\in V$ and any $c_i,c_i'\in S_i$, $C_{-i}\in  S_{-i}$, the following equality holds:
		$$
		\pi(c_i,C_{-i})-\pi(c_i',C_{-i})=u_i(c_i,C_{-i})-u_i(c_i',C_{-i}).
		$$}
\end{definition}

\begin{lemma}\label{lem4}
	The proposed RDG is an exact potential game.
\end{lemma}
\begin{proof}
    It will be proved that the following function is a potential function:
	$$
	\pi(C)=-\lambda_1\sum_{j=1}^{n}c_j^2-\lambda_2\sum_{j=1}^n(2-c_j)m_j(C).
	$$
	Denote the two terms of $\pi(C)$ as $\pi^{(1)}(C)$ and $\pi^{(2)}(C)$, respectively.
	
	Let $C=(c_i,C_{-i})$ and $C'=(c_i',C_{-i})$ be two strategy profiles before and after some $v_i$ changes its strategy from $c_i$ to $c_i'$. It can be verified  that
	\begin{align}\label{eq1}
		\pi^{(1)}(C)-\pi^{(1)}(C')&=\lambda_1 (c_i'^2-c_i^2)=g_i(C)-g_i(C').
	\end{align}
	Note that changing the status of $v_i$ does not affect the $m$-values of those vertices outside of $\bar N_i$, so one has
	\begin{align}\label{eq2}
		\pi^{(2)}(C)-\pi^{(2)}(C') & =-\lambda_2\sum_{v_j\in \bar{N_i}}(2-c_j)m_j(C)\nonumber\\
		&+\lambda_2\sum_{v_j\in \bar{N_i}}(2-c'_j)m_j(C')\nonumber\\
		& =q_i(C)-q_i(C'),
	\end{align}
	Combining \eqref{eq1} and \eqref{eq2}, gives $\pi(C)-\pi(C')=u_i(C)-u_i(C')$.	The proof is completed.
\end{proof}

As a consequence of Lemma \ref{lem4}, NEs exist. Furthermore, as shown by the following theorem, an NE can be reached in linear time steps of asynchronous interactions, where ``asynchronous'' means that players determine their strategies one by one, and thus in every round, although there are many players who are willing to change, only one actually takes action.

\begin{theorem}\label{thm0920-1}
	Starting from any initial state, the number of iterations needed for an RDG to reach an NE is at most $\frac{(4\lambda_1+2\lambda_2)n}{\min\{3\lambda_1-2\lambda_2,-4\lambda_1+3\lambda_2\}}=O(n)$.
\end{theorem}
\begin{proof}
	As long as $C$ is not an NE, there is a player $v_i$ who is willing to change his strategy from $c_i$ to $c_i'$. Denote $\Delta u_i=u_i(C')-u_i(C)$, where $C'=(c'_i,C_{-i})$. Because $v_i$ is willing to change only when he can be {\em strictly} better off. So, $\Delta u_i>0$. Then, by Lemma \ref{lem4}, $\pi(C')-\pi(C)= u_i(C')-u_i(C)>0$, that is, in every round of interaction, the potential function $\pi$ is strictly increasing. To prove the theorem, it is needed to estimate the range of $\pi$ and {\em positive lower bounds} of $\Delta u_i$ (that is, the lower bounds of $\Delta u_i$ under the condition that $\Delta u_i>0$).	
	
	Clearly, $-(4\lambda_1+2\lambda_2)n\leq \pi(C)\leq 0$. The positive lower bounds for $\Delta u_i$, in difference cases of $c_i$ and $c_i'$, are shown in Table \ref{[Tab:02]}.  For example, when $c_i=2$ and $c_i'=0$, $u_i(C)=-4\lambda_1$ by Lemma \ref{lem0} and $u_i(c_i'=0,C_{-i})=-\lambda_2\sum_{v_j\in {\bar{N}_i}}(2-c_j)m_j(C)$. Because $4\lambda_1-2\lambda_2>0$ and $4\lambda_1-3\lambda_2<0$, the {\em positive} value of $u_i(c_i'=0,C_{-i})-u_i(C)$ in this case must be $\geq4\lambda_1-2\lambda_2$. Using this argument for all cases, the estimations in  Table \ref{[Tab:02]} are obtained. Then, it can be verified that $\min\{3\lambda_1-2\lambda_2,-4\lambda_1+3\lambda_2\}$ is a positive lower bound of $\Delta u_i$. Hence, the number of iterations is at most  $\frac{(4\lambda_1+2\lambda_2)n}{\min\{3\lambda_1-2\lambda_2,-4\lambda_1+3\lambda_2\}}$, which is $O(n)$ since $\lambda_1$ and $\lambda_2$ are constants satisfying $\frac{2}{3}\lambda_2<\lambda_1<\frac{3}{4}\lambda_2$.
\end{proof}
\begin{table}[htb]
	\centering
	\caption{Estimation of the positive lower bounds of $\Delta u_i$.}
	\label{[Tab:02]}
	\begin{tabular}{rrr}
		\toprule
		$c_i$&$c_i'$&$\min\{\Delta u_i,\Delta u_i>0\}$\\
		\midrule
		$2$& $0$& $4\lambda_1-2\lambda_2$\\
		$2$& $1$& $3\lambda_1-2\lambda_2$\\
		$1$& $0$& $\lambda_1$\\
		$0$& $1$& $-\lambda_1+\lambda_2$\\
		$1$& $2$& $-3\lambda_1+3\lambda_2$\\
		$0$& $2$& $-4\lambda_1+3\lambda_2$\\
		\bottomrule
	\end{tabular}
\end{table}

\subsection{Pareto Optimality}
In the above subsection, the existence of NE and the quality of RDF corresponding to NE are discussed. In this subsection, the NE will be further analyzed from the perspective of social welfare, and it will be shown  that NE has Pareto optimality.

\begin{theorem}
	Every NE of an RDG is a Pareto-optimal solution.
\end{theorem}

\begin{proof}
	Suppose $C=(c_1,\ldots,c_n)$ is an NE but is not a Pareto-optimal solution. Then, there is a strategy profile $C'=(c_1'\ldots,c_n')\neq C$ such that
	\begin{equation}\label{eq0710-14}\centering
		\begin{split}
		\mbox{$u_i(C')\geq u_i(C)$ for any $i\in \{1,...,n\}$, and}\\
		\mbox{ $\exists j\in\{1,...,n\}$ with $u_j(C')>u_j(C)$.}
		\end{split}
	\end{equation}
	Because $C$ is an NE, by Corollary \ref{cor1217-1} and Corollary \ref{cor3}, for any $i\in\{1,\dots,n\}$, one has
	\begin{align}
		& u_i(c_i=1,C_{-i})=-\lambda_1-\lambda_2\ \mbox{or}\ -\lambda_1-2\lambda_2,\ \mbox{and} \label{eq0710-12}\\
		& u_i(c_i=0,C_{-i})=0,\ -\lambda_2,\ \mbox{or}\ -2\lambda_2.\label{eq0710-13}
	\end{align}
	
	\begin{claim}\label{claim1}
		For any $v_i$ with $c_i\neq 2$, one has $c_i'\neq 2$
	\end{claim}
	If $c_i'=2$, then $u_i(C')=-4\lambda_1$ by Lemma \ref{lem0}. By \eqref{eq0710-12} and \eqref{eq0710-13}, $u_i(C)\geq-\lambda_1-2\lambda_2$. So, $u_i(C)>u_i(C')$, contradicting condition \eqref{eq0710-14}.
	
	\begin{claim}\label{claim2}
		For any $v_i$ with $c_i=1$, one has $c_i'=1$.
	\end{claim}
	
	By Claim 1, one has $c_i'=1$ or $0$. If the claim is not true, then $c_i'=0$. Consider the two possible values for $u_i(C)$ as shown in \eqref{eq0710-12}.
	
	If $u_i(C)=-\lambda_1-\lambda_2$, then by Corollary \ref{cor1217-1}, in strategy profile $C$, $v_i$ is the unique gray vertex in $\bar N_i$ and all the vertices in $N_i$ are white. By condition \eqref{eq0710-14}, $u_i(C')\geq -\lambda_1-\lambda_2$. Thus, one must have $N_i^2(C')\neq\emptyset$, otherwise $u_i(C')\leq-2\lambda_2m_i(C')=-2\lambda_2<-\lambda_1-\lambda_2$. Let $v_k$ be a vertex in $N^2_i(C')$. Then, $u_k(C')=-4\lambda_1$ by Lemma \ref{lem0}. Since $v_k$ is white in $C$, by \eqref{eq0710-13}, one has $u_k(C)\geq -2\lambda_2$. But, then, $u_k(C)>u_k(C')$, contradicting condition \eqref{eq0710-14}.
	
	If $u_i(C)=-\lambda_1-2\lambda_2$, then by Corollary \ref{cor3}, there are exactly two vertices $v_i$ and $v_l$ in $\bar{N}_i$ that are gray in $C$ and the other vertices in $N_i\cup N_l$ are all white in $C$. By condition \eqref{eq0710-14}, one has $u_i(C')\geq -\lambda_1-2\lambda_2$. Then, there must exist $v_k\in N^2_i(C')\cup N^2_l(C')$, otherwise $u_i(C')\leq-2\lambda_2m_i(C')-\lambda_2m_l(C')=-3\lambda_2<u_i(C)$. If $v_k=v_l$, then $v_k$ is gray in $C$, and $u_k(C)=-\lambda_1-2\lambda_2$. If $v_k\neq v_l$, then $v_k$ is white in $C$, and $u_k(C)\geq-\lambda_2(m_i(C)+m_l(C))=-2\lambda_2$ by \eqref{eq0710-13}. In any case, $u_k(C)>u_k(C')=-4\lambda_1$, a contradiction.
	
	\begin{claim}\label{claim3}
		For any $v_i$ with $c_i=0$, one has $c_i'=0$
	\end{claim}
	
	If the claim is not true, then by Claim 1, one has $c_i'=1$.
	By Corollary \ref{cor1217-1} $(\romannumeral2)$, one has  $u_i(C)=-|N_i^1(C)|\lambda_2$. For any gray vertex $v_k\in N^1_i(C)$, by Claim \ref{claim2}, $v_k$ is also gray in $C'$. Since $C$ is an NE, by Corollary \ref{cor1217-1} $(\romannumeral1)$, $N^2_k(C)=\emptyset$, and thus by Claim 1, $N^2_k(C')=\emptyset$. Hence $m_k(C')=1$, and thus $u_i(C')\leq -\lambda_1-|N_i^1(C)|\lambda_2$. Consequently,  $u_i(C')-u_i(C)\leq -\lambda_1<0$, a contradiction.

	\begin{claim}\label{claim4}
		For any $v_i$ with $c_i=2$, one has $c_i'=2$
	\end{claim}
	
	Suppose $c'_i\neq2$. Because $C$ is an M-RDF by Theorem \ref{thm2}, strategy profile $(c'_i\neq2,C_{-i})$ is not an RDF. So, there is a white vertex $v_l\in N^0_i(C)$ with $N^2_l(c'_i\neq2,C_{-i})=\emptyset$. By Claim \ref{claim2} and Claim \ref{claim3}, any vertex $v_k\in \bar{N}_l$ with $k\neq i$ has $c_k'=c_k$. So, $N^2_l(C')=\emptyset$, and thus $m_l(C')=1$, $u_l(C')=-2\lambda_2-\lambda_2\sum_{v_k\in {N_l}}(2-c_k)m_k(C')$ and $u_l(C)=-\lambda_2\sum_{v_k\in {N_l}}(2-c_k)m_k(C)$. Note that $m_k(C)\leq m_k(C')$ (if $m_k(C)=1$. Therefore, $\bar N^2_k(C)=\emptyset$, and by Claim \ref{claim2} and Claim \ref{claim3}, one has $\bar N^2_k(C')=\emptyset$, which implies that $m_k(C')=1$). Then, $u_l(C')-u_l(C)\leq-2\lambda_2<0$, a contradiction.
	
	Combining the above leads to a contradiction that $C'=C$.
\end{proof}

\section{Algorithm Design and Analysis}\label{sec4}
In this section, three algorithms will be presented: {\em Game-based Asynchronous Algorithm} (GAA), {\em Game-based Synchronous Algorithm} (GSA) and {\em Enhanced Game-based Synchronous Algorithm} (EGSA).
GAA is a direct simulation of the RDG. It converges to an NE in $O(n)$ rounds of interactions, and the decision of each player  depends only on local information in its neighborhood. However, GAA is a sequential algorithm. GSA improves the efficiency of GAA by a distributed realization of the game, and thus can be better implemented by multi-agent systems. EGSA further extends GSA with the concept of {\em private contract}, which is inspired by cooperative game. It will be shown  that EGSA can achieve a better solution than GSA.

\subsection{Game-based Asynchronous Algorithm}
The pseudo code of GAA is presented in Algorithm \ref{alg1}. It is a naive simulation of the RDG. By the definition of utility functions, each player can make his decision locally: suppose each player $v_j$ stores $SS_j=(c_j,m_j(C))$ with respect to the current solution $C$. Then, each player $v_i$ can decide on his best response $BR(v_i,C)$ based on $\{SS_j\}_{v_j\in \bar N_i}$. In the algorithm, parameter $T=\frac{(4\lambda_1+2\lambda_2)n}{\min\{3\lambda_1-2\lambda_2,-4\lambda_1+3\lambda_2\}}$. By Theorem \ref{thm0920-1}, an NE can be  reached in at most this number of iterations.

\begin{figure}[h]
	\begin{algorithm}[H]
		\caption{Game-based Asynchronous Algorithm (GAA)}
		\begin{algorithmic}[1]	\label{alg1}
			\renewcommand{\algorithmicrequire}{\textbf{Input:}}
			\renewcommand{\algorithmicensure}{\textbf{Output:}}
			\REQUIRE An initial strategy profile $C^{(0)}=(c_1^{(0)},\ldots,c_n^{(0)})$
			\ENSURE An NE $C'$
			\STATE $C\leftarrow C^{(0)}$.
			\FOR {$t=1,2,\dots,T$}
			\FOR {$i=1,2,\dots,n$}
			\STATE 	$c'_i\leftarrow BR(v_i,C)$ by accessing $SS_j$ for $v_j\in\bar N_i$
			\IF {$u_i(c'_i,C_{-i})>u_i(c_i,C_{-i})$}
			\STATE	$c_i\leftarrow c_i'$
			\ENDIF
			\ENDFOR
			\IF {$C=C^{(t-1)}$}
			\STATE Break and go to line $15$
			\ELSE
			\STATE $C^{(t)}\leftarrow C$
			\ENDIF
			\ENDFOR
			\STATE  Output $C'\leftarrow C$
		\end{algorithmic}
	\end{algorithm}
\end{figure}

Note that players in GAA make decisions in a certain order. So, GAA is a sequential algorithm. To realize the game in a distributed manner, a natural idea is to let the players make decisions simultaneously. However, such a method may induce chaos to prevent the algorithm from converging. Consider the example in Fig. \ref{fig6}. Suppose the current strategy profile is $C=(0,0,0,0)$. For any player $v_i$, since $u_i(c_i'=2,C_{-i})=-4\lambda_1> -\lambda_1-5\lambda_2=u_i(c_i''=1,C_{-i})>-6\lambda_2=u_i(C)$, the best response of $v_i$ is to change 0 to 2. If all players take their best responses simultaneously, then the next strategy profile is $C'=(2,2,2,2)$. Now, for any player $v_i$, since $u_i(c_i'=0,C_{-i})=0>-\lambda_1=u_i(c_i''=1,C_{-i})>-4\lambda_1=u_i(C)$, his best response is to change from 2 to 0. A simultaneous action makes the strategy profile back to $C=(0,0,0,0)$.



To avoid such a mess, we restrict simultaneous changes to be made by a set of {\em independent} players, which leads to the GSA to be discussed in the next subsection.

\subsection{Game-based Synchronous Algorithm}\label{sec4.2}
The pseudo code of GSA is presented in Algorithm \ref{alg2}. The parameter $T$ is also $\frac{(4\lambda_1+2\lambda_2)n}{\min\{3\lambda_1-2\lambda_2,-4\lambda_1+3\lambda_2\}}$, the rationale of which is supported by Lemma \ref{lem0731-2}.
For the current strategy profile $C$, denote by $mu_i(C)=u_i(BR(v_i,C),C_{-i})-u_i(c_i,C_{-i}))$ the {\em marginal utility} of player $v_i$. In each round of iteration, all players compute their marginal utilities simultaneously, but not all those players with positive marginal utilities take actions. A player $v_i$ decides to change his strategy only when
\begin{equation}\label{eq7}
	i=\arg\min\{j:v_j\in\bar N_{i,2}\ \mbox{and}\ mu_j(C)>0\}.
\end{equation}
That is, player $v_i$ has the priority to change his strategy only when he has the smallest ID among those players in $\bar N_{i,2}$ with positive marginal utility.

\begin{figure}[h]
	\begin{algorithm}[H]
		\caption{Game-based Synchronous Algorithm (GSA)}
		\begin{algorithmic}[1]	\label{alg2}
			\renewcommand{\algorithmicrequire}{\textbf{Input:}}
			\renewcommand{\algorithmicensure}{\textbf{Output:}}
			\REQUIRE An initial strategy profile $C^{(0)}=(c_1^{(0)},\ldots,c_n^{(0)})$
			\ENSURE An NE $C'$
			\STATE $C\leftarrow C^{(0)}$.
			\FOR {$t=1,2,\dots,T$}
			\FOR {every player $v_i$ (this is done simultaneously)}
			\STATE 	$c'_i\leftarrow BR(v_i,C)$ by accessing $SS_j$ for $v_j\in\bar N_i$
			\STATE $mu_i\leftarrow u_i(c'_i,C_{-i})-u_i(c_i,C_{-i})$
			\IF {$v_i$ satisfies equation (\ref{eq7})}
			\STATE	$c_i\leftarrow c_i'$
			\ENDIF
			\ENDFOR
			\IF {$C=C^{(t-1)}$}
			\STATE Break and go to line $15$
			\ELSE
			\STATE $C^{(t)}\leftarrow C$
			\ENDIF
			\ENDFOR
			\STATE  Output $C'\leftarrow C$
		\end{algorithmic}
	\end{algorithm}
\end{figure}

The following observation shows that those players who actually change their strategies form an independent set.

\begin{observation}\label{obs2}
	Suppose $\mathcal I$ is the set of players who have actually changed their strategies simultaneously in a round of the parfor loop of Algorithm \ref{alg2}. Then, $\mathcal I$ is an independent
	set in the following sense:
	$$\bar N_i\cap\bar N_j=\emptyset,\ \forall v_i,v_j\in\mathcal I.$$
\end{observation}

In fact, if $\bar N_i\cap\bar N_j\neq\emptyset$, then $d(i,j)\leq 2$, and thus $v_i\in \bar N_{j,2}$ and $v_j\in \bar N_{i,2}$. By the rule in \eqref{eq7}, for $v_i$ and $v_j$, only the one with the smaller ID can belong to $\mathcal I$.

Intuitively, since every $v_i$ makes his decision depending only on $\bar N_i(C)$, if no other players in $\bar N_i$ change their strategies at the same time, then the decision made by $v_i$ keeps to be his best response for the altered strategy profile. So, allowing players in an independent set to change strategies simultaneously can effectively decouple mutual influences. A more detailed analysis is given in the following.

\begin{lemma}\label{lem7}
	Suppose $C$ is the current strategy profile. After one round of the parfor loop of
	Algorithm \ref{alg2}, the new strategy profile  is $C'=(C'_{\mathcal I},C_{-\mathcal I})$, where $\mathcal I$ is the set of players who have strategies changed in this round. Then
	$$\pi(C')-\pi(C)=\sum_{v_i\in\mathcal I}(u_i(c'_i,C_{-i})-u_i(c_i,C_{-i})).$$
\end{lemma}
\begin{proof}
	Suppose $\mathcal I=\{v_{i_1},v_{i_2},\dots,v_{i_t}\}$. By Observation \ref{obs2}, the player set $V$ can be decomposed to a disjoint union of sets $V=\bar N_{i_1}\cup \bar N_{i_2}\cup\ldots\cup\bar N_{i_t}\cup V_r$, where $V_r=V\setminus\bigcup_{l=1}^t\bar N_{i_l}$. Then, the potential functions $\pi(C)$ and $\pi(C')$ can be rewritten as
	\begin{equation}
		\begin{split}
		\pi(C)=&-\lambda_1\sum_{v_i\in\mathcal I}c_i^2-\lambda_2\sum_{v_i\in\mathcal I}\sum_{v_j\in\bar N_i}(2-c_j)m_j(C)\\
		&-\lambda_1\sum_{v_i\notin\mathcal I}c_i^2-\lambda_2\sum_{v_j\in V_r}(2-c_j)m_j(C),\label{eq0731-1}
		\end{split}
	\end{equation}
	\begin{equation}
	\begin{split}
	\pi(C')=& -\lambda_1\sum_{v_i\in\mathcal I}c_i'^2-\lambda_2\sum_{v_i\in\mathcal I}\sum_{v_j\in\bar N_i}(2-c_j')m_j(C')\\
	&-\lambda_1\sum_{v_i\notin\mathcal I}c_i'^2-\lambda_2\sum_{v_j\in V_r}(2-c_j')m_j(C').\nonumber
	\end{split}
	\end{equation}
	Note that $c_i'=c_i$ for any $v_i\not\in\mathcal I$. For any $v_i\in I$ and any $v_j\in \bar N_i$, by the independence of $\mathcal I$, $v_i$ is the only vertex in $\bar N_j$ whose $c$-value is changed. Since the value of $m_j$ is determined by the $c$-values in $\bar N_j$, one has $m_j(C')=m_j(c_i',C_{-i})$.
	For any $v_j\in V_r$, one has $\bar N_j\cap\mathcal I=\emptyset$. So $m_j(C')=m_j(C)$.
	Thus, $\pi(C')$ can be written as
	\begin{align}
	&\pi(C')= -\lambda_1\sum_{v_i\in\mathcal I}c_i'^2-\lambda_1\sum_{v_i\notin\mathcal I}c_i^2-\lambda_2\sum_{v_j\in V_r}(2-c_j)m_j(C)\label{eq0731-2}\\
	&-\lambda_2\sum_{v_i\in\mathcal I}\left((2-c_i')m_i(c_i',C_{-i})+\sum_{v_j\in N_i}(2-c_j)m_j(c_i',C_{-i})\right).\nonumber
	\end{align}
	The lemma follows by substracting \eqref{eq0731-1} from \eqref{eq0731-2}.
\end{proof}

\begin{theorem}\label{lem0731-2}
	Starting from any initial strategy profile $C^{(0)}$, the number of rounds for GSA \ref{alg2} to reach an NE is at most $\frac{(4\lambda_1+2\lambda_2)n}{\min\{3\lambda_1-2\lambda_2,-4\lambda_1+3\lambda_2\}}=O(n)$.
\end{theorem}
\begin{proof}
	Observe that player $v_i$ is willing to change his strategy only when $mu_i(C)>0$. By Lemma \ref{lem7}, $\pi(C')-\pi(C)=\sum_{v_i\in\mathcal I}(u_i(c'_i,C_{-i})-u_i(c_i,C_{-i}))\geq u_i(c'_i,C_{-i})-u_i(c_i,C_{-i})=mu_i(C)$. The remaining proofs are similar to that of Theorem \ref{thm0920-1}.
\end{proof}

\subsection{Enhanced Game-based Synchronous Algorithm}
Although GSA can obtain an NE in linear number of rounds of iterations, it is an S-RDF by Corollary \ref{cor1}, but the gap between an NE and a global optimal solution might be large. Consider the example in Fig. \ref{fig7}, where $C$ is an NE with $\gamma_{R}(G,C)=n-1$, and $C'$ is a G-RDF with $\gamma_{R}(G,C')=2$.

\begin{figure}[htb]
	\centering
	\subfloat[$C$, $\gamma_R(G,C)=n-1$]{\includegraphics[width=0.4\linewidth]{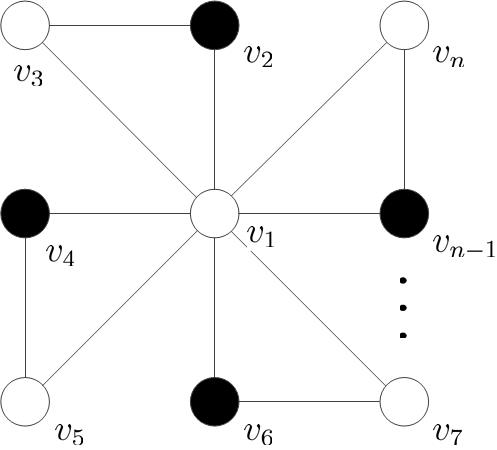}}\quad
	\subfloat[$C'$, $\gamma_R(G,C')=2$]{\includegraphics[width=0.4\linewidth]{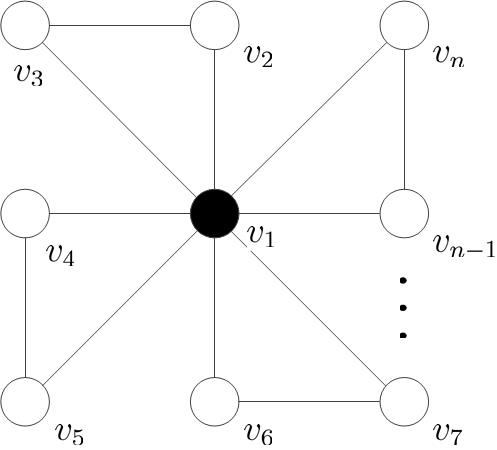}}
	\caption{An example showing that the gap between an NE and a G-RDF can be very large.}
	\label{fig7}
\end{figure}

\begin{algorithm}[thpb]
	\caption{Enhanced Game-based Synchronous Algorithm (EGSA)}
	\begin{algorithmic}[1]	\label{alg3}
		\renewcommand{\algorithmicrequire}{\textbf{Input:}}
		\renewcommand{\algorithmicensure}{\textbf{Output:}}
		\REQUIRE An initial strategy profile $C^{(0)}=(c_1^{(0)},\ldots,c_n^{(0)})$
		\ENSURE An NE $C'$
		\STATE $C\leftarrow C^{(0)}$.
		\FOR {$t=1,2,\dots,T$}
		\FOR {every player $v_i$ (this is done simultaneously)}
		\STATE 	$c'_i\leftarrow BR(v_i,C)$ by accessing $SS_j$ for $v_j\in\bar N_i$
		\STATE $mu_i\leftarrow u_i(c'_i,C_{-i})-u_i(c_i,C_{-i})$
		\IF {$v_i$ satisfies equation (\ref{eq7})}
		\STATE	$c_i\leftarrow c_i'$
		\ENDIF
		\ENDFOR
		\IF {$C=C^{(t-1)}$}
		\FOR {$i=1,2,\dots,n$}
		\IF {$c_i=0$ and $w_i(C)\geq3$}\label{line0907-1}
		\STATE $v_i$ proposes a private contract $\mathcal{A}$
		\STATE $C\leftarrow(C'_{\tau_i},C_{-\tau_i})$ (every player $v_j\in\tau_i$ agree with $\mathcal{A}$ by Lemma \ref{lem22})
		\STATE$C^{(t)}\leftarrow C$
		\STATE Break and go to line $21$
		\ENDIF
		\ENDFOR
		\ELSE
		\STATE $C^{(t)}\leftarrow C$
		\ENDIF
		\ENDFOR
		\STATE  Output $C'\leftarrow C$
	\end{algorithmic}
\end{algorithm}

There are two possible reasons for the above problems:
\begin{itemize}
	\item[(1)] The range of information that a player uses is small;
	\item[(2)] A player is assumed to be selfish, i.e., a player only maximizes his own utility regardless of the other players.
\end{itemize}

To avoid getting stuck in a bad NE such as the one in Fig. \ref{fig7} $(a)$, the players should be more cooperative, that is, there is a coalition that can make decisions simultaneously by cooperating. A large-scale coalition is intractable, a local coalition might be good enough in terms of both performance and computational complexity. For this reason, a concept of {\em private contract} is proposed as follows, which borrows the idea of contract from cooperative game theory.

\begin{definition}[Private contract]
	{\rm In the current strategy profile $C$, a private contract is proposed by a player $v_i$ which suggests that all players $v_j$ in a coalition $\tau_i\subseteq V$ (with $v_i\in\tau_i$) change their strategies from $c_j$ to $c_j'$. A private contract is {\em valid} if all players in $\tau_i$ agree with it, and the next strategy profile becomes $(C'_{\tau_i},C_{-\tau_i})$. }
\end{definition}

In the following, a private contract $\mathcal{A}$ is used, which is proposed by $v_i$, suggesting $c_i'=2$ and $c_j'=0$ for any $v_j\in\tau_i\setminus\{v_i\}$, where $\tau_i\setminus\{v_i\}=N_i^1(C)\cup\{v_j\in N_i^2(C): m_{i,j}(C)=1\}$, with
\begin{equation}\label{eq15}
	m_{i,j}(C)=\begin{cases}
		1,& \widetilde{N}_j(C)\subseteq N_i,\\
		0,& {\rm otherwise}
	\end{cases}
\end{equation}
and $\widetilde{N}_j(C)=\{v_k\in N_j^0(C): N_k^2(C)=\{v_j\}\}.$

Given a private contract proposed by $v_i$, assume that all players in coalition $\tau_i$ will agree with this contract if
\begin{equation}\label{eq0908-1}
	\Delta_{u}(C'_{\tau_i},C)=\sum_{v_j\in\tau_i}(u_j(C'_{\tau_i},C_{-\tau_i})-u_j(C))>0.
\end{equation}
The ideal is that  strictly positive utility gain $\Delta_{u}(C'_{\tau_i},C)$ can be reasonably distributed among the coalition so that each player in $\tau_i$ can get a strictly positive utility gain.  Therefore, such an assumption is reasonable.

The enhanced algorithm EGSA is described in Algorithm \ref{alg3}, where $w_i(C)$ in line \ref{line0907-1} of the algorithm is defined as follow:
$$w_i(C)=\sum_{v_j\in N_i^1(C)}m_j(C)+2\sum_{v_j\in N_i^2(C)}m_{i,j}(C).$$
It will be proved that if the condition in line \ref{line0907-1} is satisfied, then the contract proposed by $v_i$ can always be agreed. The algorithm EGSA first produces an NE $C$, which can be recognized by the criterion $C=C^{(t-1)}$. If a player $v_i$ proposes a private contract $\mathcal{A}$ that is agreed by the coalition, then $C$ is changed to $C'=(C'_{\tau_i},C_{-\tau_i})$. The algorithm continues to produce an NE from $C'$. This process is repeated until no player proposes new private contract (i.e. no player satisfies the condition in line \ref{line0907-1} of EGSA). At the termination, the output is an NE.

The following lemma describes a property that assists in proving a condition for contract $\mathcal{A}$ to be valid as well as in analyzing the time complexity of EGSA.

\begin{lemma}\label{lem8}
	In a strategy profile $C$, for any player $v_i$ with $c_i=0$ and for any player $v_j\in N_i^2(C)$, if $m_{i,j}(C)=1$, then $C'=(c_i'=2,c_j'=0,C_{-\{i,j\}})$ will not create any new white free vertex.
\end{lemma}
\begin{proof}
	If $v_j$ changes his strategy from $2$ to $0$, then all new white free vertices, if any, must belong to $\widetilde{N}_j(C)$. Because $m_{i,j}(C)=1$, one has $\widetilde{N}_j(C)\subseteq N_i$. Since $v_i$ changes his strategy from $0$ to $2$, it can strongly dominate all white vertices in $N_i\supseteq \widetilde{N}_j(C)$.
\end{proof}

The next lemma gives a sufficient condition for a private contract to be agreed by the intended coalition.

\begin{lemma}\label{lem22}
	For any NE C, the contract $\mathcal{A}$ proposed by player $v_i$ with $c_i=0$ and $w_i(C)\geq3$ is valid.
\end{lemma}
\begin{proof}
	To prove the lemma, it is needed to prove inequality \eqref{eq0908-1}.
	Let $x_1=|N_i^1(C)|$ and $x_2=|\{v_j\in N_i^2(C)\colon m_{i,j}(C)=1\}|$.
	
	Because $C$ is an NE, by Corollary \ref{cor1217-1}, every gray vertex $v_j$ is not adjacent to any black vertex, and thus $m_j(C)=1$; and every white vertex $v_k$ is adjacent to at least one black vertex, and thus $m_k(C)=0$. So, $u_i(C)=-\lambda_2x_1$. By Lemma \ref{lem0}, $u_i(C'_{\tau_i},C_{-\tau_i})=-4\lambda_1$. Hence, $u_i(C'_{\tau_i},C_{-\tau_i})-u_i(C)=\lambda_2x_1-4\lambda_1$ and $w_i(C)=x_1+2x_2$.
	
	For any vertex $v_j\in N_i^1(C)$, with similar reasons as the above, one obtains $u_j(C)=-\lambda_1-\lambda_2-\lambda_2|N_j^1(C)|$. By Lemma \ref{lem8}, there is no white free vertex produced in $C'=(C'_{\tau_i},C_{-\tau_i})$, then $u_j(C')=-\lambda_2|N_j^1(C')|$. Note that $N_j^1(C')\subseteq N_j^1(C)$. So, $u_j(C')\geq-\lambda_2|N_j^1(C)|$, and thus $u_j(C')-u_j(C)\geq\lambda_1+\lambda_2>0$.
	
	For any vertex $v_k\in N_i^2(C)$ with $m_{i,j}(C)=1$, by Corollary \ref{cor1217-1}, one has $N_k^1(C)=\emptyset$. Note that gray vertices in $C$ keeps to be gray in $C'$. Hence, $N_k^1(C')=\emptyset$. So, $u_k(C')=0$ by Lemma \ref{lem8} and $u_k(C)=-4\lambda_1$ by Lemma \ref{lem0}. Thus, $u_k(C')-u_k(C)=4\lambda_1>0$
	
	Consequently,  $\sum_{v_j\in\tau_i}(u_j(C'_{\tau_i},C_{-\tau_i})-u_j(C))\geq-4\lambda_1+\lambda_2x_1+(\lambda_1+\lambda_2)x_1
	+4\lambda_1x_2>2\lambda_1(x_1+2x_2)+\lambda_2x_1-4\lambda_1$. Because $w_i(C)\geq3$ means $x_1+2x_2\geq 3$, $\sum_{v_j\in\tau_i}(u_j(C'_{\tau_i},C_{-\tau_i})-u_j(C))\geq6\lambda_1+\lambda_2x_1-4\lambda_1=2\lambda_1
	+\lambda_2x_1>0$. The contract is valid.
\end{proof}

Next, the objective is to determine the convergence time of EGSA. For this purpose, two lemmas will be established first. The idea is as follows. When the algorithm finds an NE $C$, if a private contract proposed by $v_i$ is agreed, then the strategy profile becomes $C'=(C'_{\tau_i},C_{-\tau_i})$. Lemma \ref{lem15} shows that $\gamma_R(G,C')<\gamma_R(G,C)$. Because $\gamma_R$ can only take integer values, the $\gamma_{R}$-value is {\em strictly} decreased by at least 1. Note that except for the private contract part, EGSA is the same as GSA. Hence, by Theorem \ref{lem0731-2}, starting from $C'$, the algorithm reaches an NE $C''$ in $O(n)$ rounds. Lemma \ref{lem14} implies that $\gamma_R(G,C'')\leq \gamma_R(G,C')$. So, in the whole process, the $\gamma_R$-value is monotone non-decreasing, and in at most $O(n)$ rounds, the $\gamma_R$-value is decreased by at least 1. Since the $\gamma_R$-value is upper bounded by $2n$, the algorithm terminates in $O(n^2)$ rounds. Next, proofs are given to verify these results.

\begin{lemma}\label{lem15}
	Let $C$ be an NE. If a player $v_i$ proposes a valid private contract, then the resulting strategy profile $C'=(C'_{\tau_i},C_{-\tau_i})$ is an RDF and $\gamma_R(G,C')<\gamma_R(G,C)$.
\end{lemma}
\begin{proof}
	Since $C$ is an NE, by Theorem \ref{thm1}, it is an RDF, and thus there is no free white vertex in $C$. Note that any vertex $v_j\in \tau_i$ with $c_j=1$ is dominated by $v_i$ in $C'$ (because $c_i'=2$), and changing $c_j=1$ to $c'_j=0$ does not affect a white vertex to be strongly dominated or not (by Observation \ref{obs1}). Furthermore, by Lemma \ref{lem8}, changing the strategy of a vertex $v_j\in\tau_i$ with $c_j=2$ will not create new free white vertex. Hence, there is no free white vertex in $C'$ either, and thus $C'$ is also an RDF.
	
	Note that a player $v_i$ will propose a private contract only when $w_i(C)\geq3$. Let $x_1$ and $x_2$ be the numbers of gray vertices and black vertices in $\tau_i$, respectively. Because $w_i(C)\geq3$ implies $x_1+2x_2\geq3$, one has $\gamma_R(G,C')=\gamma_R(G,C)+2-x_1-2x_2\leq\gamma_R(G,C)-1$.
\end{proof}

\begin{lemma}\label{lem14}
	Starting from any RDF $C$, let $C'$ be the first NE reached by the algorithm after $C$. Then, $\gamma_R(G,C')\leq\gamma_R(G,C)$.
\end{lemma}
\begin{proof}
	Let $x_{ij}$ be the number of players changing their strategies from $i$ to $j$ during the iterations from $C$ to $C'$. It will be proved that
	\begin{equation}\label{eq0911-1}
		x_{01}+x_{12}+2x_{02}-x_{10}-x_{21}-2x_{20}\leq0,
	\end{equation}
	which implies the monotonicity of $\gamma_R$. Details are provided in the appendix.
\end{proof}

Call the output of EGSA as {\em enhanced Nash equilibrium} (ENE). As commented before, it is indeed an NE. Furthermore, as Lemma \ref{lem15} and Lemma \ref{lem14} imply, the size of an ENE is smaller than the size of the NE output by GSA, as long as some private contract is proposed. The time complexity of EGSA is given below, which follows from Lemma \ref{lem15} and Lemma \ref{lem14} and the argument before their proofs.

\begin{theorem}
	Starting from any initial state $C^{(0)}$, the number of rounds for EGSA to converge to an ENE is $O(n^2)$.
\end{theorem}

Let $S_{\operatorname{ENE}}$ be the set of solutions that might be output by EGSA. As the example in  Fig. \ref{fig7} indicates, not all NE belong to $S_{\operatorname{ENE}}$. Furthermore, every G-RDF is an ENE. In fact, Theorem \ref{thm4} implies that any G-RDF is an NE. If it is not an ENE, then Lemma \ref{lem15} indicates that its $\gamma$-value can be strictly decreased, contradicting the definition of G-RDF. Combining these with Corollary \ref{cor1}, the following relations are revealed, and the relationship among these sets is illustrated in Fig. \ref{fig3}.

\begin{figure*}[th]
	\centering
	\subfloat[]{\includegraphics[width=0.45\linewidth]{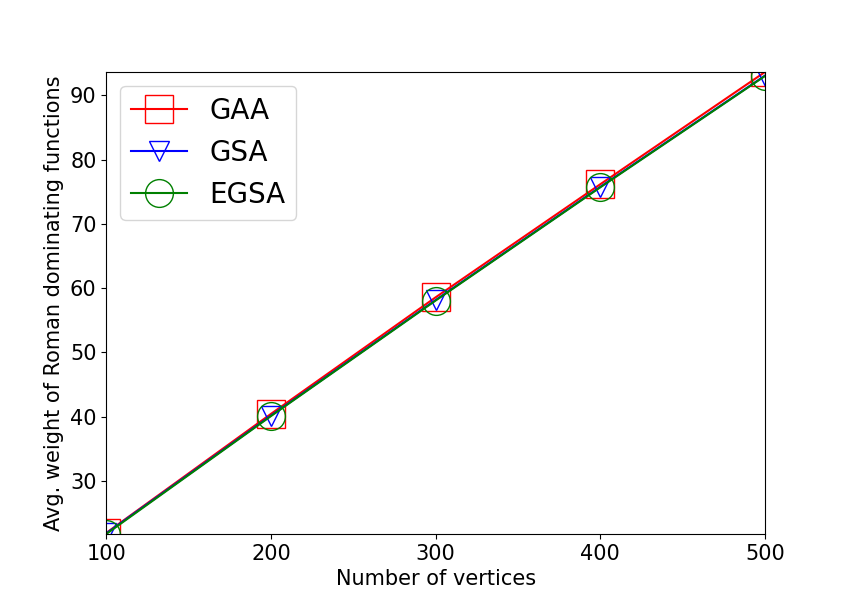}}
	\hfill
	\subfloat[]{
		\includegraphics[width=0.45\linewidth]{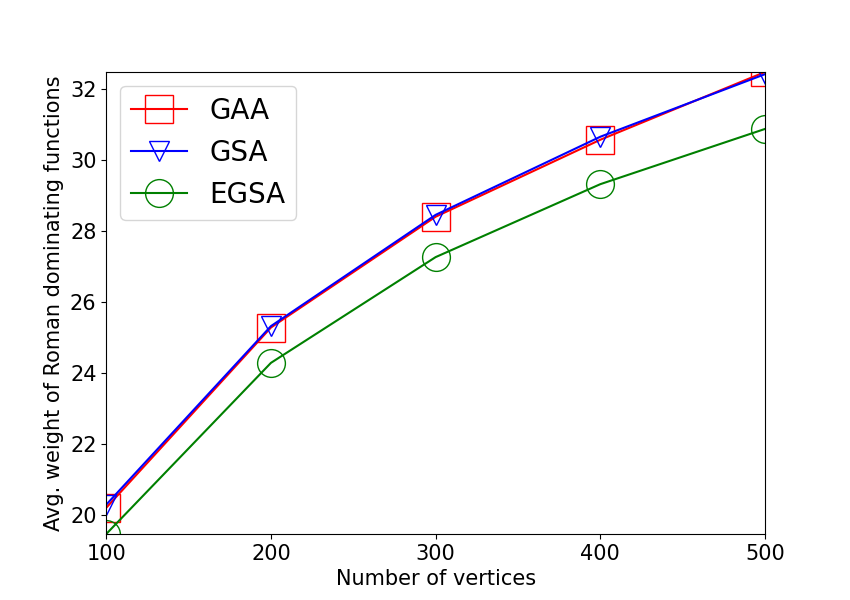}}
	\caption{Comparison of GAA, GSA and EGSA. (a) BA, $m=5$; (b) ER, $p=0.2$.}
	\label{fig8}
\end{figure*}

\begin{corollary}\label{cor2}
	In an RDG, $S_{\operatorname{G-RDF}}
	\subseteq S_{\operatorname{ENE}}
	\subseteq S_{\operatorname{NE}}
	\subseteq S_{\operatorname{S-RDF}}
	\subseteq S_{\operatorname{M-RDF}}
	\subseteq S_{\operatorname{RDF}}
	$.
\end{corollary}

\section{Simulation Results}
This section reports the  experiments on the performances of the algorithms GAA, GSA and EGSA.   All experiments are coded in Python and run with an identical configuration: AMD Ryzen 5 3500U with  Radeon Vega Mobile Gfx and 16GB of RAM.
\subsection{Comparing GAA, GSA and EGSA}

In this section, the three algorithms GAA, GSA and EGSA are compared in terms of accuracy and time complexity.
Graphs for the experiments are generated randomly using
the following two models.
\begin{itemize}
	\item [$(\romannumeral1)$] {\em The Barab\'{a}si-Albert graph} (BA) \cite{Barabasi}: Starting from a graph with a small number $m_0$ of vertices, new vertices are iteratively added. When a new vertex is added, it is connected to $m$ existing vertices, where $m\leq m_0$, and the probability that an existing vertex is linked with the new vertex is proportional to its current degree.
	
	\item [$(\romannumeral2)$] {\em The Erd$\ddot{o}$s-R\'{e}nyi graph} (ER) \cite{Erdos}: In this graph, every pair of vertexes are connected by an edge with probability $p$.
\end{itemize}

In Fig. \ref{fig8}, the horizontal axis is the number of vertices $n$ and the vertical axis shows the average weight of the solutions obtained on 1000 randomly sampled graphs. It can been seen that EGSA is better than GSA, especially in ER graph. Whereas GSA is similar to  GAA, which means that the distributed algorithm can achieve the same accuracy as the centralized algorithm. For clarity, the figures only show the situations of $m=5$ for the BA graph and
$p=0.2$ for the ER graph. In fact, for both BA and ER graphs with various parameters, all experiments show similar results.

Table \ref{[Tab:03]} shows the  average number of rounds of the three algorithms, where one round refers to one iteration of the {\em outer for loop}. Although it seems that the number of rounds of GAA is smaller than that of GSA and EGSA, it should be noted that in each round of GAA, a centralized controller has to compute $n$ players' best responses sequentially, while in each round of GSA, all players compute their best responses simultaneously. Therefore, the real time for GAA is $n$ times the number of rounds while the real time for GSA is just the number of rounds. Let $n_{GAA}$ and $n_{GSA}$ be the number of rounds in GAA and GSA, respectively. Define $\eta=\frac{n_{GSA}}{n_{GAA}\times n}$ and use it to measure the ratio in real time. The values of $\eta$ are shown in Table \ref{[Tab:04]}. It can be seen that, in terms of real time, GSA is much faster than GAA. Furthermore, it can be observed that the acceleration effect is more prominent with the increase of $n$, especially on ER.

\begin{figure*}[th]
	\centering
	\subfloat[]{\includegraphics[width=0.45\linewidth]{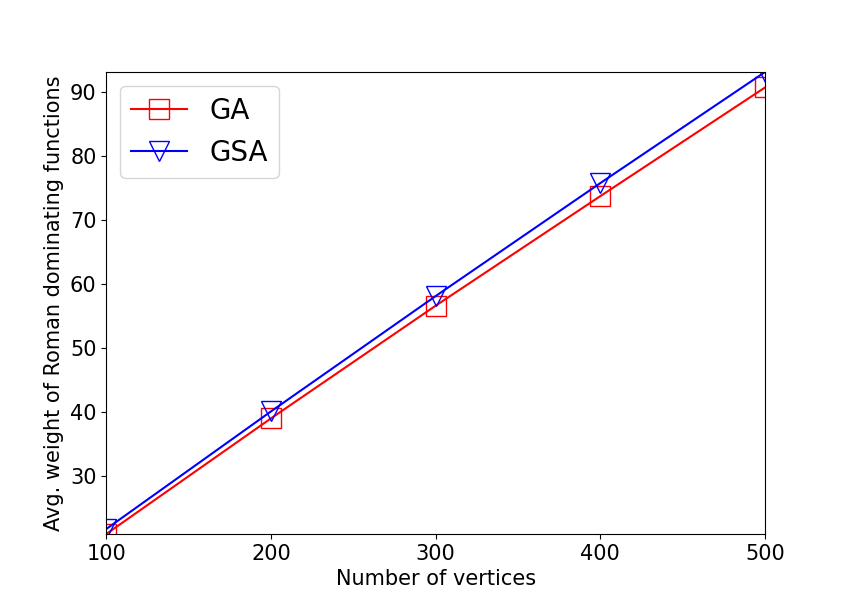}}
	\hfill
	\subfloat[]{
		\includegraphics[width=0.45\linewidth]{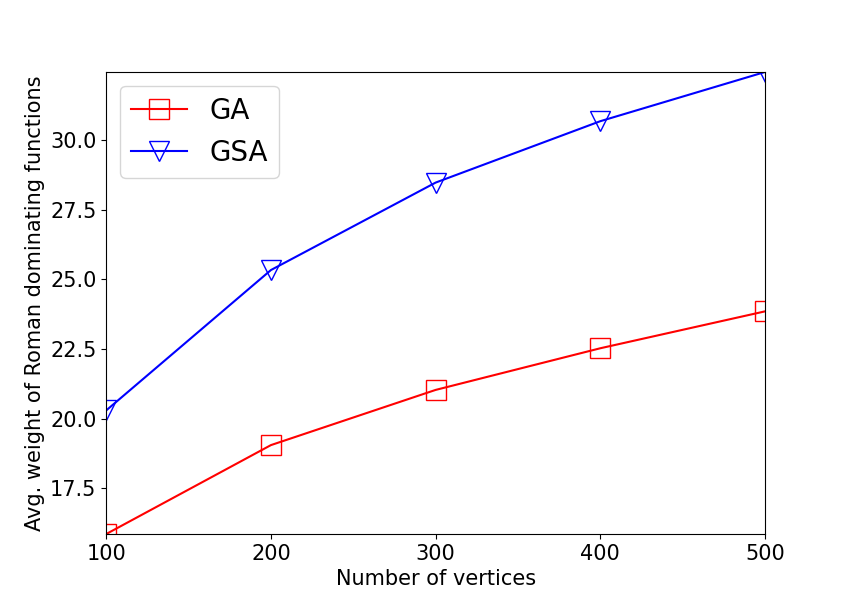}}
	\caption{Comparison of GA and GSA. (a) BA, $m=5$; (b) ER, $p=0.2$.}
	\label{fig9}
\end{figure*}

\begin{table}[htbp]
	\centering
	\caption{Average Rounds of GAA, GSA and EGSA}
	\label{[Tab:03]}
	\vskip 0.2cm
	\scalebox{0.9}{
	\begin{tabular}{|c|c|c|c|c|c|c|}
		\toprule
		Graph&Algorithm&$|V|=100$&$200$&$300$&$400$&$500$\\
		\hline
		\multirow{3}{*}{BA}& GAA& $2.332$& $2.491$& $2.615$& $2.744$& $2.879$\\
		\cline{2-7}
		&GSA& $17.968$& $30.278$& $41.199$& $50.679$& $59.314$\\
		\cline{2-7}
		& EGSA &$18.093$& $30.372$& $41.296$& $50.777$& $59.391$\\
		\hline
		\multirow{3}{*}{ER}& GAA& $2.846$& $3.062$& $3.111$& $3.190$& $3.152$\\
		\cline{2-7}
		&GSA& $22.309$& $27.267$& $29.520$& $31.369$& $32.287$\\
		\cline{2-7}
		&EGSA  & $23.364$& $28.586$& $31.031$& $33.030$& $34.241$ \\	
		\hline
	\end{tabular}}
\end{table}

\begin{table}[htbp]
	\centering
	\caption{Compare Time Complexity of GA and GSA Measured by $\eta$}
	\label{[Tab:04]}
	\vskip 0.2cm
	\begin{tabular}{|c|c|c|c|c|c|c|}
		\toprule
		Graph&$|V|=100$&$200$&$300$&$400$&$500$\\
		\hline
		BA&  $7.70\%$& $6.08\%$& $5.25\%$& $4.62\%$& $4.12\%$\\
		\hline
		ER&  $7.84\%$& $4.45\%$& $3.16\%$& $2.46\%$& $2.05\%$\\
		\hline
	\end{tabular}
\end{table}

\begin{figure*}[thbp]
	\centering
	\subfloat[]{\includegraphics[width=0.49\linewidth]{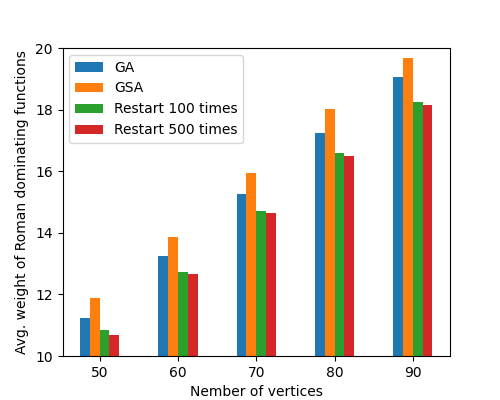}}
	\hfill
	\subfloat[]{
		\includegraphics[width=0.49\linewidth]{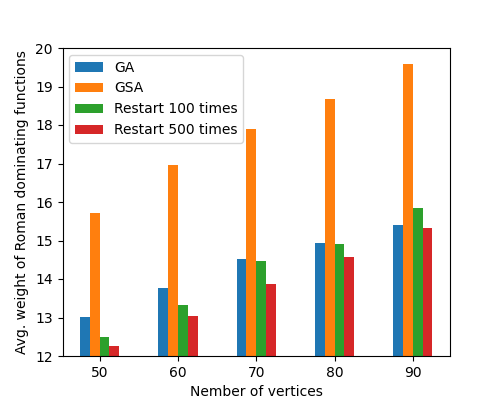}}
	\caption{Compare GA and GSA using restart. (a) BA, $m=5$; (b) ER, $p=0.2$.}
	\label{fig10}
\end{figure*}

\subsection{Comparing GSA and Greedy Algorithm}
Since the MinRD problem is NP-hard, one cannot expect a polynomial-time algorithm to obtain the exact solution. The best-known approximation algorithm for MinRD is a {\em greedy algorithm} (GA) proposed in \cite{Li3}, which can achieve  asymptotically a tight logarithmic approximation ratio.

To compare GSA with GA on BA and ER, Fig. \ref{fig9} shows the average weight of solutions computed by GSA (line with triangles) and GA (line with rectangles). In terms of weight, GA is smaller than GSA, especially on ER. This is reasonable, because GA is a {\em centralized} algorithm. The advantage of GSA is that it is a {\em distributed} algorithm. Furthermore, for the experiment in Fig. \ref{fig9}, GSA always starts from the initial strategy profile $(0,\dots,0)$. A question is: if the algorithm is restarted from different initial strategies, will it perform better? Note that Theorem \ref{lem0731-2} guarantees that GSA can converge to an NE from {\em any} initial strategy profile. Fig. \ref{fig10} shows the average weight by restarting 100 times and 500 times (also tested for restarting 200, 300 and 400 times). It can be seen that restarting GSA can achieve better solutions than GA. Note that GA is a deterministic algorithm, which cannot be benefited from restarting. It has also been tested and noticed  that the average weight will not be significantly improved with more than 100 times.

\subsection{Comparing with an Exact Solution on a Tree}
To see the accuracy of the three proposed algorithms, they are compared with the exact algorithm (DP) for MinRD on trees, which is a {\em Dynamic Program} proposed in \cite{Dreyer}. Trees for the experiments are generated in two ways as follows.
\begin{itemize}
	\item[$(\romannumeral1)$] {\em Barab\'{a}si-Albert Tree} (BAT): A tree is constructed by the BA model with $m_0=2$ and $m=1$.
	\item[$(\romannumeral2)$] {\em Random Tree} (RT):  Let $V$ be a vertex set on $n$ vertices and $F$ be the edge set consisting of all possible edges between vertices of $V$. Starting from an empty graph formed from vertex set $V$, iteratively add an edge $e$ from $F$ randomly and uniformly as long as no cycle is created, until a spanning tree on $V$ is obtained.
\end{itemize}

For each $n$, 1000 trees of size $n$ are sampled randomly. Table \ref{[Tab:06]} shows the relative error  $\omega=\frac{\gamma_R(G,f)-\gamma_R(G,f^*)}{\gamma_R(G,f^*)}$ for the RDF $f$ obtained by the algorithms (GA, GSA and EGSA) and the optimal solution $f^*$ generated by DP. Here, the results of GAA are not presented since GSA has the same accuracy as GAA, but the result of the greedy algorithm GA is included for a comparison. As can be seen from Table \ref{[Tab:06]}, GSA is superior to GA on RT. Although GA is better than GSA on BAT, with the restarting strategy, GSA is much better than GA. Furthermore, EGSA performs fairly well on trees, especially on BAT. Note that EGSA outperforms GSA even if GAS restarts 100 times.

\begin{table}[htbp]
	\centering
	\caption{Relative Errors of the Proposed Algorithms and DP Measured by $\omega$}
	\label{[Tab:06]}
	\vskip 0.2cm
	\scalebox{0.9}{
	\begin{tabular}{|c|c|c|c|c|c|c|}
		\toprule
		graph&Algorithm&$|V|=100$&$200$&$300$&$400$&$500$\\
		\hline
		\multirow{4}{*}{RT}&	GA&  $1.81\%$& $2.10\%$& $1.87\%$& $1.82\%$& $1.82\%$\\
		\cline{2-7}
		&GSA&  $1.09\%$& $1.26\%$& $1.15\%$& $1.15\%$& $1.16\%$\\
		\cline{2-7}
		&GSA(100)&  $0.24\%$& $0.41\%$& $0.60\%$& $0.80\%$& $0.92\%$\\
		\cline{2-7}
		&EGSA&  $0.12\%$& $0.37\%$& $0.31\%$& $0.33\%$& $0.32\%$\\
		\hline
		\multirow{4}{*}{BAT}&GA&  $0.17\%$& $0.14\%$& $0.19\%$& $0.19\%$& $0.18\%$\\
		\cline{2-7}
		&GSA&  $0.41\%$& $0.35\%$& $0.40\%$& $0.37\%$& $0.36\%$\\
		\cline{2-7}
		&GSA(100)&  $0.002\%$& $0.004\%$& $0.01\%$& $0.02\%$& $0.02\%$\\
		\cline{2-7}
		&EGSA&  $0.01\%$& $0.01\%$& $0.01\%$& $0.02\%$& $0.004\%$\\
		\hline
	\end{tabular}
}
\end{table}

\section{Conclusion}\label{sec6}
In this paper, we study the minimum Roman domination problem (MinRD) in a multi-agent system by game theory. A Roman domination game is proposed and the existence of Nash equilibrium (NE) is guaranteed. Furthermore, an NE can be found in a linear number of rounds of interactions. A distributed algorithm GSA is proposed to find an NE, in which every player can make a decision based on local information. It is proved that any NE is both a strong RDF and a Pareto-optimal solution. Moreover, it is shown that an enhanced NE can be obtained by an enhanced algorithm EGSA, which can further improve the quality of the solution. In the future, we will explore cooperative game theory on domination problems, hoping to obtain improved performances.

\ifCLASSOPTIONcaptionsoff
  \newpage
\fi



%

\bibliography{IEEEabrv,IEEEexample}
%

\begin{appendices}
	
	\section{Proof of Lemma \ref{lem14}}
	
	Lemma \ref{lem14} is proved by a series of lemmas below. The idea of the proofs is as follows. In both $C$ and $C'$, there is no white free vertex. It is needed to estimate, in various cases, how many white free vertices are created and diminished, resulting in an upper bound for the increase of the number of white free vertices, which is then compared with $0$. Since $C'$ is an NE, every gray vertex is free. Denote by $y$ the number of gray vertices in $C$ that are strongly dominated. Then, they will diminish in $C'$. Hence, when reaching $C'$, the total number of gray vertices that are strongly dominated is decreased by $y$. Motivated by this observation, it will then be needed to estimate the number of created and diminished gray vertices that are strongly dominated during the evolution from $C$ to $C'$, which will be compared to the total decrease with $y$. Using these relations, by some algebraic manipulation, inequality \eqref{eq0911-1} can be obtained, implying the monotonicity of $\gamma$.

	In the following, the number of changes after a strategy profile $C$ is estimated and then changed to $C'$. Note that, when a player $v_i$ changes his strategy, it only affects the statuses of those players in $\bar N_i$ (to become free or become strongly dominated). So, it is only needed to consider those vertices in $\bar N_i$.
	
	\begin{lemma}\label{lem6}
		In any strategy profle $C$, when a player $v_i$ changes his strategy from $2$ to $0$, one of the following two situations holds.
		
		\begin{itemize}
			\item [$(1)$] at most one free white vertex is created and no gray vertex becomes free;
			\item[$(2)$] at most two gray vertices become free and no free white vertex is created.
		\end{itemize}
	\end{lemma}
	\begin{proof}
		Let $C=(c_i=2,C_{-i})$ and $C'=(c'_i=0,C_{-i})$. By Lemma \ref{lem0}, $u_i(C)=-4\lambda_1$. Note that all vertices in $\bar{N}_i$ are strongly dominated by $v_i$ in $C$.
		
		If at least two white vertices in $N_i^0(C)$ become free after $v_i$ changes his strategy, then $u_i(C')\leq-4\lambda_2<-4\lambda_1=u_i(C)$, contradicting that $v_i$ is willing to change his strategy from $2$ to $0$.
		
		If at least three gray vertices in $N_i^1(C)$ become free in $C'$ after $v_i$ changes   his strategy, then $u_i(C')\leq-3\lambda_2<-4\lambda_1=u_i(C)$, a contradiction.
		
		If exactly one white vertex in $N_i^0(C)$ and at least one gray vertex in $N_i^1(C)$ become free, then $u_i(C')\leq-3\lambda_2<-4\lambda_1=u_i(C)$, a contradiction. So, if a free white (resp. gray) vertex is created, then no gray (resp. white) vertex becomes free.
		
		Combining the above arguments, the lemma is proved.
	\end{proof}
	
	The following proofs are similar, using similar notations as in Lemma \ref{lem6}.
	
	\begin{lemma}\label{lem9}
		In any strategy profile $C$, when a player $v_i$ changes his strategy from $2$ to $1$, the number of strongly dominated gray vertices is decreased by at most one and the number of white free vertices is not increased.
	\end{lemma}
	
	\begin{proof}
		Because $BR(v_i,C)=1$, one has $m_i(C')=1$ by Lemma \ref{lem5}. By the definition of $m_i(C')$, one has $N_i^2(C')=\emptyset$, and thus $N_i^2(C)=\emptyset$. As a consequence, the black vertex $v_i$ in $C$ becomes a free gray vertex in $C'$. If the number of strongly dominated gray vertices is decreased by at least two, then at least two vertices in $N_i$ that are gray in $C$ become free in $C'$, and thus $u_i(C')\leq-\lambda_1-3\lambda_2<-4\lambda_1=u_i(C)$, a contradiction. If the number of white free vertices is  increased by at least one, then at least one vertex in $N_i$ that is white in $C$ becomes free in $C'$, and thus $u_i(C')\leq-\lambda_1-3\lambda_2<-4\lambda_1=u_i(C)$, also a contradiction. These contradictions establish the lemma.
	\end{proof}
	
	For simplicity of statement, by saying the {\em status} of a white vertex or a gray vertex, it refers to the situation of this vertex being strongly dominated or free.
	
	\begin{lemma}\label{lem10}
		In any strategy profile $C$, when a player $v_i$ changes his strategy from $1$ to $0$, the number of strongly dominated gray vertices is decreased by exactly one and the number of free white vertices is not increased.
	\end{lemma}
	\begin{proof}
		Because $BR(v_i,C)=0$, one has $m_i(C')=0$ by Lemma \ref{lem5}, which implies that $v_i$ has a neighbor of $c$-value 2 in $C'$, and thus in $C$. By Observation \ref{obs1}, $m_j(C)=m_j(C')$ $\forall v_j\in V$, and thus the status of every $v_j\in V\setminus\{v_i\}$ is the same in both $C$ and $C'$. Then, the lemma follows from the fact that the strongly dominated gray vertex $v_i$ in $C$ becomes a white strongly dominated vertex in $C'$.
	\end{proof}
	
	\begin{lemma}\label{lem11}
		In any strategy profile $C$, when a player $v_i$ changes his strategy from $0$ to $1$, the number of free white vertices is decreased by exactly one and the number of gray vertices that are strongly dominated is not decreased.
	\end{lemma}
	\begin{proof}
		Because $BR(v_i,C)=1$, one has $m_i(C')=1$, and thus $N_i^2(C)=\emptyset$. Similar to the proof of the above lemma, the status of every $v_j\in V\setminus \{v_i\}$ is the same in both $C$ and $C'$. Then, the lemma follows form the fact that the free white vertex $v_i$ in $C$ becomes a free gray vertex in $C'$.
	\end{proof}
	
	\begin{lemma}\label{lem12}
		In any strategy profile $C$, when a player $v_i$ changes his strategy from $0$ to $2$, one of the following three situations holds.
		\begin{itemize}
			\item[$(1)$] The number of white free vertices is decreased by at least two and the number of strongly dominated gray vertices is not decreased;
			\item [$(2)$] The number of strongly dominated gray vertices is increased by at least three and the number of free white vertices is not increased;
			\item[$(3)$] The number of free white vertices is decreased by exactly one and the number of strongly dominated gray vertices is increased by at least one.
		\end{itemize}
	\end{lemma}
	\begin{proof}
		If in $C$, there are at most two free gray vertices in $N_i^1(C)$ and no free white vertex in $\bar N_i^0(C)$, then $u_i(C)\geq -2\lambda_2>-4\lambda_1=u_i(C')$, a contradiction. So, there are at least three free gray vertices or at least one free white vertex that become strongly dominated in $C'$. In other words, one of the following two situations occurs: $(a)$ the number of strongly dominated gray vertices is increased by at least three, or $(b)$ the number of free white vertices is decreased by at least one.
		
		Note that changing $c_i=0$ to $c_i'=2$ will never increase the number of free white vertices. So, if $(a)$ occurs, then one has situation $(2)$.
		
		Next, consider the case when $(b)$ occurs. If the number of free white vertices is decreased by exactly one and the number of strongly dominated gray vertices is not increased, then there is exactly one free white vertex in $\bar{N}_i^0(C)$ and no free gray vertex in $N_i^1(C)$. It follows that $u_i(C)=-2\lambda_2>-4\lambda_1=u_i(C')$, a contradiction. So, if the number of free white vertices is decreased by exactly one, then one has situation $(3)$.
		
		If the number of free white vertices is decreased by at least two, then changing $c_i=0$ to $c_i'=2$ will never decrease the number of strongly dominated gray vertices, so one has situation $(1)$.
	\end{proof}
	
	\begin{lemma}\label{lem13}
		In any strategy profile $C$, when a player $v_i$ changes his strategy from $1$ to $2$, one of the following three situations holds.
		\begin{itemize}
			\item[$(1)$] The number of free white vertices is decreased by at least one and the number of strongly dominated gray vertices is not decreased;
			\item[$(2)$] The  number of strongly dominated gray vertices is increased by at least two and the number of free white vertices is not increased.
			\item[$(3)$] The number of free white vertices is decreased by at least two and the number of strongly dominated gray vertices is decreased by exactly one.	
		\end{itemize}
	\end{lemma}
	\begin{proof}
		If there are at most two free gray vertices in $\bar N^1_i(C)$ and no free white vertex in $N^0_i(C)$, then $u_i(C)\geq-\lambda_1-2\lambda_2>-4\lambda_1=u_i(C')$, a contradiction. So one of the following two situations holds: $(a)$ there are at least three free gray vertices in $\bar N^1_i(C)$, or $(b)$ there are at least one free white vertex in $N^0_i(C)$.
		
		If $(a)$ occurs, then the number of strongly dominated gray vertices is increased by at least two (note that if $v_i$ is free, since it becomes black in $C'$, it cannot be counted as a new strongly dominated gray verex). Combining this with the fact that changing $c_i=1$ to $c_i=2$ will never increase the number of free white vertices, one has situation $(2)$.
		
		Next, suppose $(b)$ occurs. Note that if $v_i$ is a strongly dominated gray vertex in $C$ and the number of free white vertices is decreased by exactly one, then  $u_i(C)=-\lambda_1-2\lambda_2>-4\lambda_1=u_i(C')$, a contradiction. Also note that changing $c_i=1$ to $c_i'=2$ will not create new free gray vertex. Thus, the number of strongly dominated gray vertex can be decreased by at most one, and this decrease occurs only when $v_i$ is a strongly dominated gray vertex in $C$. Combining these observations, if $v_i$ is a strongly dominated gray vertex in $C$, then one has situation $(3)$. If $v_i$ is a free gray vertex in $C$, then one has situation $(1)$.
	\end{proof}
	
	Now, one is ready to prove Lemma \ref{lem14}.
	
	\begin{proof}[ Proof of Lemma \ref{lem14}]
		In the process of evolution from $C$ to $C'$, where $C$ and $C'$ are the strategy profiles specified by the condition of Lemma \ref{lem14}, let $x_{ij}$ be the number of players who have changed their strategies from $i$ to $j$. Next, it will be proved that
		\begin{equation}\label{eq0911-3}
			x_{01}+x_{12}+2x_{02}-x_{10}-x_{21}-2x_{20}\leq0.
		\end{equation}
		
		In Lemma \ref{lem6}, there are two situations. Use $x_{20}^1$ and $x_{20}^2$ to represent the numbers of times that situation $(1)$ and situation $(2)$ occurred, respectively. Then, $x_{20}=s_{20}^1+x_{20}^2$. Similarly, $x_{21}=x_{21}^1$, $x_{10}=x_{10}^1$, $x_{01}=x_{01}^1$. Note that the three situations in Lemma \ref{lem12} and the three situations in Lemma \ref{lem13} may have some overlap, so $x_{02}\leq x_{02}^1+x_{02}^2+x_{02}^3$ and $x_{12}\leq x_{12}^1+x_{12}^2+x_{12}^3$.
		
		By Lemma \ref{lem6} to Lemma \ref{lem13}, the number of white free vertices created in the process is at most $x_{20}^1-x_{01}^1-2x_{02}^1-x_{02}^3-x_{12}^1-2x_{12}^3$, and the number of strongly dominated gray vertices in the process is decreased by at most $2x_{20}^2+x_{21}^1+x_{10}^1-3x_{02}^2-x_{02}^3-2x_{12}^2+x_{12}^3$.
		
		Because $C'$ is an NE, by Corollary \ref{cor1217-1}, there is no white free vertex in $C'$ and any gray vertex in $C'$ is free. Because $C$ is an RDF, there is no white free vertex in $C'$. Let $y$ be the number of strongly dominated gray vertices in $C$. Then,
		\begin{equation}\label{eq9}
			x_{20}^1-x_{01}^1-2x_{02}^1-x_{02}^3-x_{12}^1-2x_{12}^3\geq0,
		\end{equation}
		\begin{equation}\label{eq10}
			2x_{20}^2+x_{21}^1+x_{10}^1-3x_{02}^2-x_{02}^3-2x_{12}^2+x_{12}^3\geq y
		\end{equation}
		Adding (\ref{eq10}) and (\ref{eq9}), one has
		\begin{equation}
		\begin{split}
			x_{20}^1+2x_{20}^2+x_{21}^1+x_{10}^1-x_{01}^1-2x_{02}^1\\
        	-3x_{02}^2-2x_{02}^3-x_{12}^1-2x_{12}^2-x_{12}^3\geq y
		\end{split}
		\end{equation}
		Combining this with the relationship between $x_{ij}$ and $x_{ij}^k$, one obtains
		$$
		x_{01}+x_{12}+2x_{02}-x_{10}-x_{21}-2x_{20}\leq -y-x_{20}^1-x_{02}^2-x_{12}^2.
		$$
		Since the variables $y$ and $x_{ij}^k$'s are all non-negative, inequality \eqref{eq0911-3} is proved, which completely verifies Lemma \ref{lem14}.
	\end{proof}
\end{appendices}
\end{document}